\documentclass[11pt]{article}
\usepackage{amssymb, latexsym, mathtools, amsthm}
\begingroup
    \makeatletter
    \@for\theoremstyle:=definition,remark,plain\do{%
        \expandafter\g@addto@macro\csname th@\theoremstyle\endcsname{%
            \addtolength\thm@preskip\parskip
            }%
        }
\endgroup
\usepackage{enumerate}
\usepackage{pgf,tikz}
\usepackage{pdfsync}
\usepackage{txfonts}
\usepackage[T1]{fontenc}
\usepackage{booktabs}
 \usepackage{graphicx}
\definecolor{dnrbl}{rgb}{0,0,0.3}
\definecolor{dnrgr}{rgb}{0,0.3,0}
\definecolor{dnrre}{rgb}{0.5,0,0}
\usepackage[colorlinks=true, citecolor=dnrgr, linkcolor=dnrre, urlcolor=dnrre] {hyperref}
\usepackage{xcolor}
\usepackage{parskip}
\usepackage{tabularx, colortbl}

\usepackage{geometry}
 \usepackage[scriptsize, up]{caption}
\usepackage{pgf,tikz}
\usetikzlibrary{decorations, decorations.pathmorphing}
\usetikzlibrary {shapes}

\theoremstyle{plain}
\newtheorem{thm}{Theorem}[section]

\newtheorem{prop}[thm]{Proposition}
\newtheorem{lem}[thm]{Lemma}
\newtheorem{coro}[thm]{Corollary}
\newtheorem{defi}[thm]{Definition}

\numberwithin{equation}{subsection}
\makeatletter

\let\c@table\c@figure
\makeatother

\newcommand{\Nat}{\mathbb{N}}

\newcommand{\restr}{\upharpoonright}  
\newcommand{\un}{\uparrow} 
\newcommand{\de}{\downarrow} 

\DeclarePairedDelimiter{\tuple}{\langle}{\rangle}

\newcommand{\DHMN}{Downey, Hirschfeldt, Miller, and Nies\ }

\newcommand{\BG}{Becher and Grigorieff \ }
\newcommand{\BFGM}{Becher, Figueira, Grigorieff, and Miller\ }
\newcommand{\KS}{Ku{\v{c}}era and Slaman\ }
\newcommand{\CHKW}{Calude, Hertling, Khoussainov, and Wang\ }
\newcommand{\DHN}{Downey, Hirschfeldt, and Nies\ }

\renewcommand{\DH}{Downey and Hirschfeldt\ }

\newcommand{\ml}{Martin-L\"{o}f }
\newcommand{\pz}{$\Pi^0_1$\ }

\newcommand{\eg}{e.g.\ }
\newcommand{\ie}{i.e.\ }
\newcommand{\ce}{c.e.\ }

\newcommand{\lce}{left-c.e.\ }

\newcommand{\rce}{right-c.e.\ }

\newcommand{\pf}{prefix-free }

\renewenvironment{abstract}
 { \normalsize
  \list{}{
    \setlength{\leftmargin}{.0cm}%
    \setlength{\rightmargin}{\leftmargin}%
    }%
  \item {\bf \abstractname.} \relax}
 {\endlist}

\newcommand{\rcep}{$\mathbf{0}'$-right-c.e.\ }
\newcommand{\lcep}{$\mathbf{0}'$-left-c.e.\ }

\newcommand{\lcepp}{$\mathbf{0}^{(2)}$-left-c.e.\ }
\DeclarePairedDelimiter{\dbra}{\llbracket}{\rrbracket}
\DeclarePairedDelimiter{\extp}{\mathtt{EXT}(}{)} 
\DeclarePairedDelimiter{\dom}{\mathtt{DOM}(}{)} 
\DeclarePairedDelimiter{\fin}{\mathtt{FIN}(}{)} 
\DeclarePairedDelimiter{\comp}{\mathtt{COM}(}{)} 
\DeclarePairedDelimiter{\cofin}{\mathtt{COF}(}{)} 
 
\newcommand{\rceppp}{$\mathbf{0}^{(3)}$-right-c.e.\ }
\newcommand{\lceppp}{$\mathbf{0}^{(3)}$-left-c.e.\ }

  \makeatletter
\newtheorem*{rep@theorem}{\rep@title}
\newcommand{\newreptheorem}[2]{%
\newenvironment{rep#1}[1]{%
 \def\rep@title{#2 \ref{##1}}%
 \begin{rep@theorem}}%
 {\end{rep@theorem}}}
\makeatother
 \newreptheorem{thm}{Theorem}
\newreptheorem{lem}{Lemma}

\title{Random numbers as probabilities of machine behaviour 
\thanks{This research was partially supported by NSF DMS-1362273. Barmpalias was also supported by the 
1000 Talents Program for Young Scholars from the Chinese Government, grant no.\ D1101130,
the Chinese Academy of Sciences (CAS) and the Institute of Software of the CAS.
We thank Veronica Becher and the referees for many corrections and suggestions that improved the presentation of this article.}}
\author{George Barmpalias  \and Douglas Cenzer \and Christopher P.~Porter}
\date{\today}
\begin{document}
\maketitle
\begin{abstract}
A fruitful way of obtaining meaningful, possibly concrete,
algorithmically random numbers is to consider a potential behaviour of a Turing machine
and its probability with respect to a measure (or semi-measure) on the input space of 
binary codes. 
In this work we obtain characterizations of the
algorithmically random reals in higher randomness classes, 
as probabilities of certain events that can happen when an oracle
universal machine runs probabilistically on a random oracle. Moreover we apply our analysis
to several machine models, including oracle Turing machines, 
prefix-free machines, and models for infinite online computation.
We find that in many cases the arithmetical complexity of a property is directly reflected in the 
strength of the algorithmic randomness of the probability with which it occurs, on any given universal machine.
On the other hand, we point to many examples where this does not happen and the probability is
a number whose algorithmic randomness is not the maximum possible (with respect to its arithmetical complexity).
Finally we find that, unlike the halting probability of a universal machine, the probabilities of 
more complex properties like totality, cofinality, computability or completeness 
do not  necessarily 
have the same Turing degree when they are defined with respect to different universal machines.
\end{abstract}
\vspace*{\fill}
\noindent{\bf George Barmpalias}\\[0.5em]
\noindent
State Key Lab of Computer Science, 
Institute of Software, Chinese Academy of Sciences, Beijing, China.
School of Mathematics, Statistics and Operations Research,
Victoria University of Wellington, New Zealand.\\[0.2em] \textit{E-mail:} \texttt{barmpalias@gmail.com}.
\textit{Web:} \texttt{\href{http://barmpalias.net}{http://barmpalias.net}}\par\medskip
\noindent{\bf Douglas Cenzer}\\[0.5em]
\noindent Department of Mathematics
University of Florida, Gainesville, FL 32611\\[0.2em]
\textit{E-mail:} \texttt{cenzer@math.ufl.edu.}
\textit{Web:} \texttt{\href{http://people.clas.ufl.edu/cenzer}{http://people.clas.ufl.edu/cenzer}}\par\medskip
\noindent{\bf Christopher P.~Porter}\\[0.5em]
\noindent Department of Mathematics and Computer Science,
Drake University,
Des Moines, IA 50311\\[0.2em]
\textit{E-mail:} \texttt{cp@cpporter.com.}
\textit{Web:} \texttt{\href{http://cpporter.com}{http://cpporter.com}}\par
\vfill \thispagestyle{empty}
\clearpage
\section{Introduction}
In this work we examine the probabilities of certain outcomes of a universal Turing machine that
runs on random input from the point of view of algorithmic complexity.
The input of a Turing machine can be a finite binary string that is written on the input tape, or even
an infinite binary stream whose bits are either written on the input tape at the start of the computation or
are provided upon demand as the computation progresses.  Such computations
can be interpreted in terms of probabilistic machines and randomized computation or as
oracle computations with a random oracle.
The study of probabilistic Turing machines goes back to \cite{deleeuw1955}, where
it was shown that if a randomized machine computes a function $f$ with positive probability,
then $f$ is computable. Similarly, if  a randomized machine produces an enumeration of
a set $W$ with positive probability, then $W$ is computably enumerable,  \ie $W$ can be enumerated
by a deterministic machine without any oracle.

One way to introduce algorithmic randomness to the study of randomized computation  is to
consider what a machine can do given an algorithmically random oracle of a certain strength.
This line of research has produced a large body of work and a sub-discipline in the area between
algorithmic randomness and computability; see \cite{typical,typvol} for a presentation and a 
bibliography. A common theme
in this topic is that algorithmically random oracles of sufficient strength cannot compute useful
sets or functions. For example, oracles that are random relative to the halting problem 
do not compute any 
complete extension of Peano Arithmetic (Stephan \cite{MR2258713frank}).

In this paper we apply the theory of algorithmic randomness to randomized computation in a different way.
Given a property of the outcome of a randomized computation, we consider the probability
with which this property occurs on a given universal machine.
Then we consider the algorithmic randomness of this probability as a real number in $[0,1]$.
It turns out that in this way we can characterize the real numbers that are probabilities of certain key
properties (like halting, totality, computability, etc.),  as the algorithmically random numbers of certain
well known classes in algorithmic information theory.

\subsection{Previous work on the topic and outline of our results in historical context}\label{TIA78CsuO}
One can trace this methodology back to Chaitin \cite{MR0411829} where it was shown
that the probability that a universal self-delimiting 
machine halts is a \ml random real. Becher, Daicz, and Chaitin \cite{firstBC}
applied this idea to a model for infinite computations which was 
introduced earlier in Chaitin \cite{CHAITIN1976233} and
exhibited a probability which is 2-random, \ie \ml random relative to the halting problem.
Becher and Chaitin \cite{fuin/BecherC02} studied a related model for infinite computations in order to
exhibit a probability which is 3-random, \ie \ml random relative to 2-quantifier sentences in arithmetic. 
A subsequent series of papers by Becher and Grigorieff
\cite{DBLP:journals/jsyml/BecherG05,tcs/BecherG07,DBLP:journals/jsyml/BecherG09}
as well as Becher, Figueira, Grigorieff and Miller \cite{jsyml/BecherFGM06}, 
followed this idea with the main aim of exhibiting algorithmically random numbers that are concrete,  
in the sense that they can be expressed as probabilities that a universal machine will have a certain
outcome when it is run on a random input.
Sureson \cite{Sureson20151} continues this line of research, showing that many natural open sets
related to the universal machine have random probabilities of different algorithmic strength, while
his arguments appeal to completeness phenomena as opposed to machine arguments.

A crucial development  in the study of halting probabilities was
the cumulative work of 
Solovay \cite{Solovay:75u}, \CHKW \cite{Calude.Hertling.ea:01}, and \KS \cite{Kucera.Slaman:01}.
In these papers it was established that a real number is the halting probability of a universal 
self-delimiting machine if and only if
it is \ml random and has a 
computably enumerable (in short, c.e.) left Dedekind cut.
Moreover, as we explain in the following, this analysis shows that the same characterization remains
true for most commonly used types of Turing machines and is not specific to the self-delimiting model
which is sometimes used in algorithmic information theory.
A similar approach was followed in \cite{Barmpalias3488}
in order to provide a characterization of the {\em universality probability} of a self-delimiting machine,
a notion that was introduced in \cite{Wallace99minimummessage}.

In the present paper we take the aforementioned work as a starting point and develop
a methodology for characterizing algorithmically random numbers as probabilities of
certain outcomes of a universal machine like totality, computability, and co-finiteness.
There are three main ways that our work stands out from previous attempts on this topic. 
We provide

\begin{enumerate}[\hspace{0.3cm}(1)]
\item characterizations of algorithmically random numbers as probabilities;
\item results that hold uniformly for several types of Turing machine models;
\item examples of probabilities which are {\em not} as random as their
arithmetical complexity would suggest;
\item demonstrations that the probability of properties that are not definable with one quantifier is not
Turing degree-invariant with respect to different universal machines.
\end{enumerate}

We have a few remarks for each of these clauses.
With respect to (1), the only such characterizations in the literature 
are the two that we mentioned, namely the halting probability and the universality probability.
However, even in these cases the characterizations were proved for the specific model of
self-delimiting machines although, as we will see, they apply in a much more general class of models.
With respect to (2) we would like to note that
all examples of random probabilities from the literature that we have encountered refer to
self-delimiting models and there is a reason for this preference. Self-delimiting models
correspond to prefix-free domains, which means that the probabilities considered
turn out to be measures of open sets, which are easily represented and manipulated as sets of strings.
In the present paper we free ourselves from this restriction, considering probabilities that are measures
of classes that have higher Borel complexity, so they are not necessarily open or closed or even 
$G_{\delta}$. Moreover, as we will see, this generality does not introduce an extra burden in our 
analysis. The reason for this pleasing fact can be traced to the theory of 1-random reals with
left Dedekind cut (also known as \lce reals) and the fact that the measures of 
$\Sigma^0_n$ classes are the same as the measures of $\Sigma^0_1$ classes
relative to $\mathbf{0}^{(n-1)}$, \ie the Turing degree of the halting problem iterated $n-1$ times.
As a result of this generality, we can look at classic properties of computability theory like
totality, cofinality, completeness or computability, and characterize their probabilities in various standard
models of Turing machines.

Clause (3) deserves a somewhat more elaborate remark.
There is a direct contrast between the definability of a mathematical object  and 
its algorithmic complexity in terms of effective randomness. Indeed, according to \ml 
\cite{MR0223179} a sequence  is random if it is not contained in a null set of a certain
arithmetical complexity. So, for instance, since each arithmetically definable sequence $x$ belongs to the null set $\{x\}$, which has a arithmetical definition with the same complexity as that of $x$, $x$ cannot be random with respect to null sets of this particular arithmetic complexity. When one looks at properties of the universal machine, it often turns out that
the probabilities of the occurrence of these properties are, in a sense, maximally random with respect to the arithmetical class that they
belong to. Such is the case with the halting probability of a universal \pf machine, 
also known as Chaitin's $\Omega$, which is a real with a $\Sigma^0_1$ left Dedekind cut and is
{\em 1-random} \ie random with respect to all \ml tests, a notion based on $\Sigma^0_1$ definable open sets.
Indeed, a real with  $\Sigma^0_1$ left Dedekind cut cannot be random with respect to stronger statistical tests
such as \ml tests relative to the halting problem, or even random with respect to null $\Pi^0_2$ classes
(which give a slightly stronger notion of randomness than \ml randomness).
Grigorieff, as discussed in \cite{jsyml/BecherFGM06}, conjectured  that
in a certain more restricted context, the randomness of a property of a universal machine
has strength corresponding to the arithmetical class that it belongs to. This conjecture, in
the specialized form that it appeared in, was refuted
by \BFGM in \cite{jsyml/BecherFGM06}, where many examples of properties of every level
of the  arithmetical hierarchy were given which have probabilities which are not even random relative
to the halting problem. Later work by \BG
\cite{DBLP:journals/jsyml/BecherG05,tcs/BecherG07,DBLP:journals/jsyml/BecherG09}
salvages this conjecture but in specific contexts
(using specialized machines or non-standard notions of universality). In Section \ref{rT1udReGjP}
we point to a rather common but nontrivial 
reason why a property of a universal machine may fail to have
maximal algorithmic complexity with respect to its arithmetical class.

Regarding (4), recall that the halting probabilities with respect to the various universal
Turing machines all have the same Turing degree, namely the degree of the halting problem. In
other words, given the halting probability of one universal machine as an oracle, we can compute
the halting probability of any other universal machine. Thus we may say that the (universal) 
halting probability is degree-invariant with respect to the choice of the underlying machine.
This is no longer true when one considers probabilities with more complex arithmetical complexity.
This was demonstrated by Barmpalias and Dowe \cite{Barmpalias3488} who studied
the {\em universality probability}, a notion from  \cite{Wallace99minimummessage} which is 
only definable with four alternating quantifiers, and showed that
its Turing degree depends on the choice of the underlying machine.
In the present paper we give several more examples of this phenomenon and observe that it typically
occurs when the arithmetical 
complexity of the probability is above the complexity of the halting problem.

\subsection{Random versus algorithmically random}
The notion of randomness has several interpretations in different fields of science, \eg in probability theory,
in algorithmic randomness, and in quantum physics. This paper has to do with the first two of these uses,
as we examine the algorithmic randomness of certain probabilities associated with
running a machine probabilistically on an oracle chosen at random with respect to the uniform distribution
on the binary outcomes of a hypothetical experiment.
Hence there are two different ways that we use the word `random' in this paper. The first one is {\em random}
as in probability theory, \ie an object drawn at random from a given distribution (in our case, the uniform
distribution on a binary outcome). The second one is {\em algorithmically random} in the sense
of algorithmic randomness, \ie referring to a fixed individual object which may or may not be
{\em algorithmically} random with respect to a given definition of algorithmic randomness.

The first meaning of the word `random' is indicated by writing the word in italics
(like {\em  random}) and its use has been restricted with reference to the oracle of a
Turing machine. When the word is used in the sense of algorithmic randomness, it always
refers to the probability of a given outcome of a randomized machine 
as a real number. 
For these reasons there should be no confusion regarding the two
uses of the word `random' in the following discussion. 
In addition to this caveat, we note that in formal statements regarding randomness of 
probabilities, the word `random' appears with a qualification indicating the type or strength 
of algorithmic randomness in question,
\eg \ml random or 2-random. 

\subsection{Turing machines}\label{IzttMzGQSQ}
We do not aim at a comprehensive classification of all properties of universal machines on all models of computation.
However our methodology applies rather uniformly to a number of standard models of computation 
and properties of these models. In order to demonstrate this fact, we have chosen 
three models of interest and examine the probabilities that certain key properties
from classical computability theory occur when computation is performed probabilistically on
a {\em random} input stream, which can also be thought of as an oracle.
In this paper we focus our attention to the following models of computation:
\begin{enumerate}[\hspace{0.5cm}(1)]
\item Oracle (Turing) machines computing (partial) functions.
\item Monotone machines computing strings or streams.
\item Chaitin's infinitary self-delimiting machines, computing strings, or streams.
\end{enumerate}

In case (1) we consider a standard notion of computation of a (partial) function $f:\Nat\to\Nat$ from 
a Turing machine with a {\em random} oracle. Such a computation may be infinitary, in the sense that
we compute a potential infinite object $f$, but it is a juxtaposition of countably many finite computations,
namely the computation of each $f(n)$ for $n$ in the domain of $f$. 
In this case we can ask more complex
question about the output $f$, for example if the domain of $f$ has a certain property or if $f$ is total.
Universality of oracle machines is defined in a  standard way.
A set of strings is called \pf if it does not contain two strings such that one is a proper 
extension of the other. Moreover we let $\ast$ denote concatenation of strings and let
$\simeq$ denote the fact that either the expressions on either side of it are undefined
(or do not halt) or both of these expressions are defined and are equal.

\begin{defi}[Universal oracle machines]
Given an effective list $(M_e)$ of all oracle machines, an 
oracle machine $U$ is universal if 
there exists a computable function $e\mapsto\sigma_e$ from numbers into a \pf set of strings, such that $U(\sigma_e\ast X,n)\simeq M_e(X,n)$ for all $e, X,n$.
\end{defi}

In case (2) we consider  
{\em monotone machines}, which 
were used by Levin in \cite{levinthesis,Levin:73} in order to give a definition of
the algorithmic complexity of finite objects.  
Let $\preceq$ denote the prefix relation between strings. 

\begin{defi}[Monotone machine]
A monotone machine is a 
Turing machine $M$ with input/output finite binary strings and the 
monotonicity property that
if $\sigma\preceq \tau$ and $M(\sigma), M(\tau)$ both halt, then $M(\sigma)\preceq M(\tau)$.
\end{defi}
A monotone machine $M$ can be regarded as model for infinitary computations, where
$M(X)$, the output on an infinite input $X$, is defined as the supremum of all $M(\sigma)$
for all $\sigma\prec X$.
In this context we can ask whether the output is finite or infinite, or even if the output belongs
to a certain class (\eg if it has a tail of 1s).
The universality of monotone machines is defined in a  standard way.  

\begin{defi}[Universal monotone machines]
Given an effective list $(M_e)$ of all monotone machines, an 
monotone machine $U$ is universal if 
there exists a computable function 
$e\mapsto\sigma_e$ from numbers into a \pf set of strings, such that
$U(\sigma_e\ast\tau)\simeq M_e(\tau)$ for all $e, \tau$.
\end{defi}

We remark that these standard notions of universality are quite different from the notion
of optimal machines in the context of Kolmogorov complexity (e.g.\ see \cite[Definition 2.1]{MR1438307}).


Finally in case (3) we consider a specialized 
self-delimiting model for infinite computations, which was introduced
by Chaitin in \cite{CHAITIN1976233} (we defer the formal definition of this model to Section 
 \ref{ELZgFIxGxP},
see Definition \ref{7osrvJOuxR}). 
In Section \ref{ELZgFIxGxP} we will see that this model can be seen as an infinitary \pf
machine with an infinitary notion of halting which is \pz instead of
the usual $\Sigma^0_1$ halting that one considers in finitary computations in Turing machines.
For this reason, Chaitin's self-delimiting infinitary machine model
is  different from the models in cases (1) and (2).
We consider it because it has been discussed in several papers, as we pointed out in 
Section \ref{TIA78CsuO}.

\subsection{Characterizations of the probabilities of key outcomes of randomized machines}\label{skotQ7rkRJ}
Recall that by the
cumulative work of 
Solovay \cite{Solovay:75u}, \CHKW \cite{Calude.Hertling.ea:01}, and \KS \cite{Kucera.Slaman:01}, the
halting probabilities of universal \pf machines are exactly the 1-random \lce reals.
Let $\emptyset'$ denote the halting problem and let $\mathbf{0}'$ denote the
Turing degree of $\emptyset'$. More generally,
let $\emptyset^{(n)}$ denote the halting problem iterated $n$ times and let
$\mathbf{0}^{(n)}$ denote the Turing degree of $\emptyset^{(n)}$.
A real is called $\mathbf{0}^{(n)}$-left-c.e.\ if its left Dedekind cut is \ce relative to $\mathbf{0}^{(n)}$.
Similarly, 
a real is called $\mathbf{0}^{(n)}$-right-c.e.\ 
if its right Dedekind cut is \ce relative to $\mathbf{0}^{(n)}$.
In Barmpalias and Dowe \cite{Barmpalias3488}, it was shown that the universality probabilities 
(a notion from \cite{Wallace99minimummessage}) of
universal \pf machines are exactly the 4-random \rceppp reals.

 In this section we present our main results, which are characterizations of probabilities
 of certain machine outcomes in the models presented in Section \ref{IzttMzGQSQ}.
 We start with the following characterization of 
2-random \rcep reals in terms of the {\em totality probability} of
 oracle Turing machines (case (1) in Section \ref{IzttMzGQSQ}).

\begin{thm}[Totality for oracle machines]\label{YdyxkMqxRX}
The probability that the function computed by a universal oracle machine is total 
is a 2-random \rcep real. Conversely every 2-random \rcep real in (0,1) is the probability
that a certain universal oracle  machine computes a total function.
\end{thm}
We will also show that Theorem \ref{YdyxkMqxRX} also holds with regard to the 
probability that {\em the function computed by a universal oracle Turing machine has infinite domain}.
Moreover, the same argument shows a similar characterization of the probability that a universal
monotone machine (case (1) in Section \ref{IzttMzGQSQ}) has infinite output, as we note below. 
This probability for
an infinitary self-delimiting machine (case (3) in Section \ref{IzttMzGQSQ}) 
has the same characterization, but the proof is different,
as it depends on the definition of this type of infinitary computation.

\begin{thm}[Infinite output of monotone  or infinitary self-delimiting  machines]\label{7MGSCRK2Yz}
The probability that the output of a universal monotone machine is infinite 
is a 2-random \rcep real. Conversely every 2-random \rcep real in $(0,1)$ is the probability
that a certain universal monotone machine computes a total function. Moreover the same is true
for Chaitin's self-delimiting infinitary machines.
\end{thm}
Naturally, in order to obtain characterizations of higher randomness we need to
consider properties of higher arithmetical complexity.
Recall that a subset of $\Nat$ is cofinite if its complement is finite. 
Moreover recall that the indices of partial computable functions with a cofinite domain
is a $\Sigma^0_3$-complete set.  The same is true of the
partial computable functions with a computable domain.
This is an indication that these notions can give the required characterizations of 3-randomness.

Before we demonstrate this intuition, we show that one has to be careful to chose the right model.
As we explain in Section \ref{ELZgFIxGxP}, 
there is a notion of domain of an infinitary self-delimiting machine,
and properties of randomized computations with respect to this model correspond to subsets
of the domain of the machine in question. In this sense, the following lemma says that we cannot
exhibit 3-randomness by considering properties of computations in this model.

\begin{lem}[Limits on randomness of outcomes of infinitary self-delimiting machines]\label{5JCUWCcstk}
The measure of any subset of the domain of an infinitary self-delimiting machine is not a 
3-random real number.
\end{lem}
In particular, the property that the output of a universal infinitary self-delimiting machine 
has a tail of 1s (\ie is the characteristic
sequence of a cofinite subset of $\Nat$) does not have 3-random probability. The same is true
when one considers monotone machines, but for different reasons as will be seen in Section \ref{rT1udReGjP}.

\begin{prop}[Limits on randomness of computable outcomes of oracle and monotone machines]\label{XqR4J6nv5M0}
The probability that an oracle machine computes the characteristic sequence of a cofinite
subset of $\Nat$ is not a 3-random real number. The same is true for monotone machines.
\end{prop}

The following result concerning
oracle machines (case (1) in Section \ref{IzttMzGQSQ})  
demonstrates the correct notion of cofiniteness and computability which gives the
desired characterization of 3-randomness.

\begin{thm}[Cofiniteness  and computability for oracle machines]\label{Sgl7O4jKzk}
The probability that a universal oracle  machine computes a function with cofinite domain is
a 3-random \lcepp real. Conversely, every 3-random \lcepp real in $(0,1)$ is the probability that
a certain universal oralce machine computes a function with cofinite domain.
The same holds for `computable' instead of `cofinite'.
\end{thm}

Since the domain of a partial computable function can be seen as the range of another
partial computable function, we will see that Theorem \ref{Sgl7O4jKzk} also holds
for the {\em range} instead of the {\em domain} of functions.
Finally, we note that similar methods can be used in order to obtain the following characterization
of 4-randomness. 

\begin{thm}[Hardness for oracle machines]\label{7FXYwDy4zY}
The probability that a universal oracle machine computes a function whose domain computes the
halting problem is
a 4-random \lceppp real. Conversely, every 4-random \lceppp real in $(0,1)$ is the probability that
a certain universal oracle machine computes a function whose domain computes the
halting problem.
\end{thm}
The proof that the index set of Turing complete \ce sets is $\Sigma^0_4$-complete
plays a crucial role in the proof of Theorem \ref{7FXYwDy4zY}, just as we use 
the $\Sigma^0_3$-completeness of the index set of the cofinite and  
computable \ce sets in order to obtain Theorem
\ref{Sgl7O4jKzk}.

\subsection{Turing degrees of probabilities of universal machines}
The characterization of 1-random \lce reals as halting probabilities of universal \pf machines
 by Solovay \cite{Solovay:75u}, \CHKW \cite{Calude.Hertling.ea:01}, and \KS \cite{Kucera.Slaman:01}
relativizes to an arbitrary oracle. In particular, $n$-random $\mathbf{0}^{(n-1)}$-left-c.e.\ reals
are exactly the halting probabilities of \pf machines with oracle $\mathbf{0}^{(n-1)}$.
Hence, given our characterizations of Section \ref{skotQ7rkRJ}, the probabilities we consider can
be viewed as halting or non-halting probabilities  of universal \pf machines relative to $\mathbf{0}^{(n-1)}$
for some $n>0$.
We will combine these facts with the work of \DHMN
\cite{MR2188515} on relativizations of Chaitin's halting probability
in order to show the following dependence of the degree of the probabilities that we considered.

\begin{coro}[Turing degree variance]\label{FziaqEGYYr}
Given any machine model from (1)-(3) of Section \ref{IzttMzGQSQ} and any 
probability of that model considered in 
Theorems \ref{YdyxkMqxRX}, \ref{7MGSCRK2Yz}, \ref{Sgl7O4jKzk},
and \ref{7FXYwDy4zY}, the Turing degree of the particular probability depends on the underlying
universal machine that is chosen.
In fact, in each of these cases, if $P_M$ denotes the particular probability with respect
to a machine $M$ of the given type, there exist universal machines $U,V$ such that 
$P_U$ and $P_V$ have incomparable Turing degrees.
\end{coro}

Indeed, \DHMN
\cite{MR2188515} showed that for each $n>0$ there exist oracle universal \pf machines $U,V$
such that the degrees of $\Omega_U^{\emptyset^{(n)}}$, $\Omega_V^{\emptyset^{(n)}}$
are incomparable. This fact, along with the above discussion and the results of 
Section \ref{skotQ7rkRJ} give Corollary \ref{FziaqEGYYr}.
The reader who would like a more detailed discussion of the
relativization of the characterization of halting probabilities and the application of
the results from \cite{MR2188515} is referred to
\cite{Barmpalias3488} where a similar implication was established, 
along with a detailed presentation of the necessary background.

More can be said about the Turing degrees of the probabilities we discussed in Section  
\ref{IzttMzGQSQ}. By \DHMN
\cite{MR2188515}, for each oracle $X$ and each universal oracle \pf machine $U$ 
we have
$\Omega_U^{\emptyset^{(n)}}\oplus\emptyset^{(n)}\equiv_T \emptyset^{(n+1)}$ for all $n>0$.
Moreover, for all $n>0$, the Turing degrees of $\Omega_U^{\emptyset^{(n)}}$ and $\emptyset^{(n)}$ form
a minimal pair. Hence if $V$ is an oracle machine and $T_V$ is the probability that $T$ computes
a total function, then the Turing degree of $T_V$ forms a minimal pair with the degree of the 
halting problem while $T_V\oplus \emptyset'$ computes $\emptyset^{(2)}$. Similar facts can be deduced
about all the probabilities considered in the theorems of Section \ref{IzttMzGQSQ}.

\section{Background material, terminology and notation}
A standard reference regarding the interaction of 
computability theory and algorithmic randomness are
\DH \cite{rodenisbook} and Nies \cite{Ottobook}, while Calude \cite{caludebook} has an information-theoretic perspective. 
Odifreddi \cite{Odifreddi:89,Odifreddi:99} is a standard reference in classical computability
theory while Li and Vit{\'a}nyi \cite{MR1438307} is a standard reference in the theory of Kolmogorov complexity.

 The arguments in this paper
borrow many ideas from the theory of arithmetical complexity and in particular
 $\Sigma^0_n$-completeness. A compact presentation of this 
part of computability theory can be found in \eg \cite[Chapter IV]{MR882921}.
A more comprehensive source on this topic is \cite{Hinman1978-HINRH}.
In this section we introduce some notions that are directly related to our analysis,
and prove some related results which will be used in the main part of the paper.

Recall the types of machines that we discussed in Section \ref{IzttMzGQSQ}.
If $M$ is an oracle  machine then $M(X)$ refers to 
a (partial) function $n\mapsto M(X,n)$ from $\Nat$ to $\Nat$. 
Let $\mathtt{TOT}(M)$ denote the class of all $X$ such that $M(X)$ is total.
We occasionally refer to the measure of  $\mathtt{TOT}(M)$ as the 
{\em totality probability} of $M$. 
Let $\cofin{M}$ consist of the reals $X$ 
such that the domain of $M(X)$ is cofinite.
The measure of $\cofin{M}$ is sometimes referred to as the 
{\em cofiniteness probability} of $M$.
Moreover we let $\mathtt{INF}(M)$
denote the class of all $X$ such that the domain of $M(X)$ is infinite.
We refer to the measure of $\mathtt{INF}(M)$ 
as the {\em infinitude probability} of $M$.

A monotone machine $N$ is viewed as a
(partial) function $\sigma\mapsto N(\sigma)$ from strings to strings 
with a monotonicity property (see Section \ref{IzttMzGQSQ}).
Moreover in this case $N(X)$ denotes the supremum
of all $N(\sigma)$ where $\sigma$ is a prefix of $X$.
Given a monotone machine $N$ we let $\mathtt{INF}(N)$
be the set of all $X$ such that $N(X)$ is infinite, \ie a stream.
Similarly, let
$\mathtt{FIN}(N)$
be the set of all $X$ such that $N(X)$ is finite, \ie a string.
We refer to the measure of $\mathtt{INF}(N)$ 
as the {\em infinity probability} of $N$.
Also let $\cofin{N}$ denote the reals $X$ in
$\mathtt{INF}(N)$ such that $N(X)$
is equal to $\sigma\ast 1^{\omega}$ for some binary string $\sigma$.
Finally, let $\comp{N}$ denote the set of reals $X$ such that the domain of $N(X)$  is computable. 

Similar notations and terminology 
apply to the infinitary self-delimiting machines, which we introduce
in Section \ref{ELZgFIxGxP}.
A string $\sigma$ is {\em compatible} with a string $\tau$ if $\sigma\preceq\tau$ or 
$\tau\preceq\sigma$; otherwise we say that  
$\sigma$ is {\em incompatible} with $\tau$. As it is customary in computability theory,
the suffix $[s]$ on a parameter which is part of a construction that takes place in countably
many stages indicates
the value of the parameter at stage $s$.

\subsection{Arithmetical classes and measure}\label{pm7oworwdR}
We work with subsets of the Cantor space, the class of infinite binary sequences 
with the usual topology generated by the basic open sets
$\dbra{\sigma}=\{X\ |\ \sigma\prec X\}$ where $\prec$ denotes the (strict) prefix relation.
Given a set of strings $S$, define 
$\dbra{S}$ to be the set of all reals that have a prefix in $S$. We use $\ast$ to denote concatenation
of strings.
Moreover if $P$ is a set of reals and $\sigma$ is a string, then $\sigma\ast P$ denotes the
set of reals $Y$ of the form $\sigma\ast X$, where $X\in P$.
We consider the uniform Lebesgue product measure on subsets of the Cantor space.
Recall that a class of reals is $\Sigma^0_n$ if it is definable by a
$\Sigma^0_n$ formula in arithmetic. Moreover 
$\Pi^0_n$ classes are the complements of $\Sigma^0_n$ classes.
These definitions relativize to any oracle $X$, obtaining the $\Sigma^0_n(X)$ and the  
$\Pi^0_n(X)$ classes which are definable in arithmetic with an extra parameter $X$.
The $\Sigma^0_1$ classes are known as the effectively open sets. However note that
for $n>1$ the $\Sigma^0_n$ classes are not necessarily open, so they are not necessarily
$\Sigma^0_1(\emptyset^{(n-1)})$.
The following fact will be used several times in our analysis.
In the first part of Lemma \ref{TLBE5A9gU4} the qualification {\em uniformly} means that from
an index of a  $\Sigma^0_n$ class we can compute an index for a $\mathbf{0}^{(n-1)}$-computable
increasing sequence of rationals converging to the measure of the class.
Similar meaning is intended by  the qualification {\em uniformly}  in the second part of the lemma,
but with respect to decreasing approximations by rationals.

\begin{lem}[Measures of arithmetical classes]\label{TLBE5A9gU4}
Let $n>0$. 
The measure of a  $\Sigma^0_n$ class is uniformly a $\mathbf{0}^{(n-1)}$-left \ce real. 
Similarly,
the measure of a  $\Pi^0_n$ class is  uniformly a $\mathbf{0}^{(n-1)}$-right \ce real. 
\end{lem}
\begin{proof}
By induction on $n$. For $n=1$ this clearly holds. Assume that the statement is true for $n=k$
and consider a $\Sigma^0_{k+1}$ class $S$. Then $S=\cup_i P_i$ where $(P_i)$ is a uniform
increasing sequence of $\Pi^0_{k}$ classes. 
By the hypothesis, the reals $\mu(P_i)$ are uniformly \rce relative to $\mathbf{0}^{(k-1)}$.
So the reals $\mu(P_i)$ are uniformly $\mathbf{0}^{(k)}$-computable. On the other hand, $\mu(S)$ is
the supremum of all $\mu(P_i)$. Hence $\mu(S)$ is \lce relative to $\mathbf{0}^{(k)}$. A similar
argument shows that the measure of a $\Pi^0_{k+1}$ class is \rce relative to $\mathbf{0}^{(k)}$.
\end{proof}
Hence, although $\Sigma^0_n$ classes are topologically very different from the
$\Sigma^0_1(\emptyset^{(n-1)})$ classes for $n>1$, their measures are the same type of 
real numbers. This simple fact will be used in our main analysis.
Note that $\Sigma^0_1(\emptyset^{(n-1)})$ classes can be represented by
$\mathbf{0}^{(n-1)}$-\ce sets of binary strings. We denote the set of binary strings by $2^{<\omega}$
and the set of binary streams by $2^{\omega}$.
Note that $\Sigma^0_1$ sets of strings are exactly the \ce sets of strings.

For $\Sigma^0_2$ upward closed sets of strings we will use the following fact:

\begin{lem}[Canonical $\Sigma^0_2$ approximations]
If $U$ is a $\Sigma^0_2$ upward closed set of strings,
there exists
a computable sequence $(V_s)$ of finite sets
of strings  with upward closures $\mathbf{V}_s$
respectively such that
\begin{enumerate}[\hspace{0.5cm}(i)]
\item for each string $\sigma$ we have $\sigma\in U$ if and only if there exists $s_0$ 
such that $\sigma\in \mathbf{V}_s$ for all $s>s_0$;
\item there are infinitely many $s$ such that $\mathbf{V}_s\subseteq U$.
\end{enumerate}
We call $(V_s)$  a {\em canonical $\Sigma^0_2$ approximation to $U$} (although it is not an approximation in the literal sense).
\end{lem}

\begin{proof}
Given a set of strings $S$, let $\mathbf{S}$ be the upward closure of $S$. Given $U$,
there exists a \ce operator $W$ such that $W^{\emptyset'}=U$. 
We can modify the enumeration of $W$ (obtaining a modified $\widehat{W}$) 
with respect to a computable enumeration
$(\emptyset'_s)$ of $\emptyset'$, so that if $n\in \emptyset'_{s+1}-\emptyset'_s$ for some number $n$ 
and stage $s$, any number $m$ which is in $W^{\emptyset'_{s+1}}_{s+1}$ with oracle use above $n$
is not counted in $\widehat{W}^{\emptyset'_{s+1}}_{s+1}$. In other words, the enumeration of 
$\widehat{W}^{\emptyset'}$ follows the enumeration of  $W^{\emptyset'}$ except that it delays the enumeration
of certain numbers until a stage where the associated segment of $\emptyset'$ remains stable between the
current and the previous stages. This is a standard technique which is known as the
{\em hat-trick} (originally due to Lachlan) and is applied to functionals and \ce operators relative to \ce sets 
(see Odifreddi \cite[Section X.3]{Odifreddi:99} for an extended discussion on this method).
Let $V_s=\widehat{W}^{\emptyset'_{s}}_{s}$. 
A {\em true} enumeration into $\emptyset'$ is an enumeration of a number $n$ into $\emptyset'$ at stage $s$
such that $\emptyset'_{s}\restr_n=\emptyset'\restr_n$. A {\em true stage} is a stage $s$ at which 
a true enumeration occurs. Clearly there exist infinitely many true enumerations and stages. 
By the choice of $\widehat{W}$ we have that if $s$ is a true stage then  
$V_s\subseteq U$. Since $U$ is upward closed, we also have $\mathbf{V}_s\subseteq U$ for infinitely many 
stages $s$. Moreover, if $\sigma\in U$ then $\sigma\in V_s\subseteq \mathbf{V}_s$ for all but finitely many stages $s$. Finally, by the hat trick, if
$\sigma\in \mathbf{V}_s$ for all but finitely many stages $s$, we necessarily have $\sigma\in U$
because $\mathbf{V}_t\subseteq U$ for infinitely many stages $t$.
\end{proof}

The converse of Lemma \ref{TLBE5A9gU4} is also true in a strong effective way.
This is a consequence of the Kraft-Chaitin theorem (\eg see \cite[Section 2.6]{rodenisbook}) 
which says that, given a
computable sequence of positive integers $(c_i)$ such that $\sum_i 2^{-c_i}\leq 1$,
we can effectively produce a computable sequence of binary strings $(\sigma_i)$
such that $|\sigma_i|=c_i$ for each $i$ and the set $\{\sigma_i\ |\ i\in\Nat\}$ is prefix-free.
Moreover the Kraft-Chaitin theorem holds relative to any oracle.
Recall that $\Sigma^0_1(\emptyset^{(n-1)})$ classes of reals $P$ can be represented by
$\Sigma^0_n$ sets $S$ of strings, in the sense that $P=\dbra{S}$. 
The measure of a set of binary strings $S$ is defined to be the measure
of $\dbra{S}$. Moreover these
sets of strings can be chosen to be prefix-free.
\begin{lem}[Measures of arithmetical classes, converse]\label{PMFeQmge4u}
Let $n>0$. Given any $\mathbf{0}^{(n-1)}$-left-\ce real $\alpha\in [0,1]$, we can effectively produce a
$\Sigma^0_n$ \pf set of strings of measure $\alpha$. Similarly, if $\beta\in [0,1]$ is a
$\mathbf{0}^{(n-1)}$-right-\ce real, we can effectively produce a 
$\Pi^0_n$ \pf set of strings of measure $\beta$.
\end{lem}
\begin{proof}
By symmetry it suffices to prove the first statement. 
Note that if $\alpha=0$ then the statement is trivial, so assume that $\alpha\in (0,1]$
and that $\alpha$ is a $\mathbf{0}^{(n-1)}$-left-\ce real.
Then there exists a $\mathbf{0}^{(n-1)}$-computable sequence of positive integers $(c_i)$
such that $\alpha=\sum_i 2^{-c_i}$.
By the Kraft-Chaitin theorem relative to $\mathbf{0}^{(n-1)}$ we can obtain a
$\mathbf{0}^{(n-1)}$-computable sequence of binary strings $(\sigma_i)$ 
such that the set $S:=\{\sigma_i\ |\ i\in\Nat\}$ is
a prefix-free set and $|\sigma_i|=c_i$ for each $i$. Therefore $\mu(\dbra{S})=\alpha$ 
and this concludes the proof.
\end{proof}
Lemma \ref{TLBE5A9gU4} and a modified version of Lemma \ref{PMFeQmge4u} will be used
in the proofs of most of the results presented in Section \ref{skotQ7rkRJ}.

\subsection{\ml randomness}\label{lkNQieYqH}
We use the formulation of algorithmic randomness in terms of effective statistical tests,
as introduced by \ml in \cite{MR0223179}.
A \ml test is a uniformly \ce sequence of $\Sigma^0_1$ classes $(U_i)$ such that
$\mu(U_i)\leq 2^{-i}$ for each $i$. More generally, a \ml test relative to an oracle $X$ is
a  uniformly \ce relative to $X$ sequence of $\Sigma^0_1(X)$ classes $(U_i)$ such that
$\mu(U_i)\leq 2^{-i}$ for each $i$. 
A real $\alpha$ is \ml random (also known as 1-random) if $\alpha\not\in \cap_i U_i$
for every \ml test $(U_i)$. More generally, for each $n$ we say that a real $\alpha$ is $(n+1)$-random
if it is \ml relative to $\mathbf{0}^{(n)}$, \ie $\alpha \not\in \cap_i U_i$
for every \ml test $(U_i)$ relative to $\mathbf{0}^{(n)}$.
\ml \cite{MR0223179} showed that there are {\em universal} \ml tests $(U_i)$, in the sense that
the reals that are not 1-random are exactly the reals in $\cap_i U_i$.
A similar statement is true regarding $n$-randomness.
It is well known, \eg from \cite{Kucera.Slaman:01}, that the measure $U_i$ of any member of a universal
\ml test is a 1-random real. Similarly, the measure of any member of a universal \ml test
relative to $\mathbf{0}^{(n)}$ is $(n+1)$-random.

Demuth \cite{Dempseu} showed that if $\alpha$ is a 
1-random \lce real and $\beta$ is another \lce real then $\alpha+\beta$ is a 1-random \lce real.
\DHN \cite{Downey02randomness} showed that
conversely, if $\alpha,\beta$ are \lce reals and  $\alpha+\beta$ is 1-random, then at least one of 
$\alpha,\beta$ is 1-random. These results relativize to an arbitrary oracle.
In particular, if  $\alpha$ is  
$n$-random and a $\mathbf{0}^{(n-1)}$-left-\ce 
real and $\beta$ is another $\mathbf{0}^{(n-1)}$-left-\ce 
real then $\alpha+\beta$ is $n$-random and is
a $\mathbf{0}^{(n-1)}$-left-\ce real.
\DHN \cite{Downey02randomness} also showed that
if $\alpha$ is a 1-random \lce real  and $\beta$ is any \lce real then  there exists some $c_0$ such that
$\alpha-2^{-c}\cdot \beta$ is a \lce real for each $c>c_0$. 
This fact relativizes and also applies to \rce reals. In particular, for each $n>0$,
\begin{equation}\label{5zlG2HCdb}
\parbox{12cm}{if $\alpha,\beta$ are $\mathbf{0}^{(n-1)}$-left-\ce reals 
and $\alpha$ is $n$-random then there exists some $c_0\in\Nat$ such that 
$\alpha-2^{-c}\cdot \beta$ is a $\mathbf{0}^{(n-1)}$-left-\ce real
for each $c>c_0$.}
\end{equation}
Moreover an analogous statement holds for $\mathbf{0}^{(n-1)}$-right-\ce reals.
This result is heavily based on earlier work by
Solovay \cite{Solovay:75u} and \CHKW \cite{Calude.Hertling.ea:01}.

\section{Oracle machines and monotone machines}\label{oQfiYDwopG}
In this section we prove Theorem \ref{YdyxkMqxRX} and the part of Theorem
\ref{7MGSCRK2Yz} which refers to monotone machines.
The plan of the proof is as follows. 
We first show how to construct a machine of the desired type, whose totality probability is
2-random. Then we argue that because of this first result, universal machines of the given type
have 2-random totality probabilities. The final step is to prove the converse, \ie that given any
2-random \lcep real number $\alpha\in(0,1)$, 
we can find a universal machine of the given type, whose totality
probability is $\alpha$. We note that the arguments involved in the three steps we just described
are interconnected in an essential way. In particular, in order to prove the converses 
of Theorems \ref{YdyxkMqxRX} and \ref{7MGSCRK2Yz} we need to use the lemmas
established while showing the randomness of the totality probability. 

\subsection{Randomness of totality and infinitude probabilities}\label{6iJldN1dR6}
We start with showing how to build an oracle machine $M$ such that the set of reals 
$X$ for which $n\mapsto M(X,n)$ is a total function is based on a given $\Sigma^0_2$ set of strings.
Recall from Section \ref{IzttMzGQSQ} that if $M$ is an oracle machine then $M(X)$ refers to 
a (partial) function $n\mapsto M(X,n)$ from $\Nat$ to $\Nat$. 
 
\begin{lem}[From a $\Sigma^0_2$ set to an oracle machine]\label{w1HmrXQxym}
Given any $\Sigma^0_2$ set $U$ of strings, there exists
an oracle  machine $M$ such that for each real $X$ the function $M(X)$ is total 
if and only if it has infinite domain
if and only if $X$ does not have a prefix in $U$.
\end{lem}
\begin{proof}
Let  $(U_s)$ be a canonical $\Sigma^0_2$ approximation to $U$ as defined in 
Section \ref{pm7oworwdR}. 
We define the oracle machine $M$ in stages as a partial computable function from 
$2^{<\omega}\times\Nat$ to $\Nat$.

Let $M_0$ be the empty machine. At stage $s+1$, for each string $\sigma$ of length at most $s$
such that $M_s(\sigma,|\sigma|)$ is not defined, check if $\sigma$ has a prefix in $U_{s+1}$. If it does not,
then define $M_{s+1}(\sigma,|\sigma|)=0$. This concludes the construction.

First, note that $M$ is a well defined oracle  machine since given $\sigma\prec\tau$
and $n$ such that $M(\sigma,n)$, $M(\tau,n)$ are defined we have $M(\sigma,n)=M(\tau,n)=0$.
Moreover the oracle use is the identity function, regardless of the oracle.
Suppose that $X$ does not have a prefix in $U$ and let $n$ be a number. Since $X\restr_n$ is not in
$U$, there will be a stage $s+1>n$ such that $X\restr_n\not\in U_{s+1}$. Hence by the construction
we have $M(X\restr_n,n)=0$. Hence if $X$ does not have a prefix in $U$ then $M(X)$ is total.
Conversely, suppose that $X$ does have a prefix $X\restr_n$ in $U$. There there is a stage $s_0$
such that $X\restr_n\in U_{s+1}$ for all $s>s_0$. Hence by the construction, for all $m>s_0$
we have that $M(X\restr_{m},m)$ is undefined. Again by the construction, this means that
for all $m>s_0$ we have that the domain of $M(X)$ is finite, so $M(X)$ is not total. 
Lastly, if $M(X)$ is not total then $X$ has a prefix in $U$, the latter implies that
$M(X)$ has finite domain, and this last clause implies that $M(X)$ is not total. This concludes the proof. 
\end{proof}

Recall from Section \ref{IzttMzGQSQ} that a monotone machine $N$ is viewed as a
(partial) function $\sigma\mapsto N(\sigma)$ from strings to strings with the monotonicity property
that if $\sigma\preceq \tau$, $N(\sigma)\de$, and $N(\tau)\de$, then $N(\sigma)\preceq N(\tau)$.
Moreover in this case $N(X)$ denotes the supremum
of all $N(\sigma)$ where $\sigma$ is a prefix of $X$.

\begin{lem}[From a $\Sigma^0_2$ set to  an oracle machine]\label{qEkDLqgp25}
Given a $\Sigma^0_2$ set $U$ of strings, there exists
a monotone machine $N$ such that for all reals $X$ the output $N(X)$ is infinite 
if and only if $X$ does not have a prefix in $U$.
\end{lem}
\begin{proof}
Consider the machine $M$ of Lemma \ref{w1HmrXQxym} and for each string
$\sigma$ let $N(\sigma)$ be $0^{|\sigma|}$ if $M(\sigma,i)\de$ for all $i<|\sigma|$
and undefined otherwise. Then clearly $N$ is a monotone machine. Moreover since the use
function of the oracle machine $M$ is the identity, for each $X$ we have that $M(X)$ is a total function
if and only if $N(X)$ is infinite. Then the desired properties of $N$ follow from the properties of $M$
according to Lemma \ref{w1HmrXQxym}.
\end{proof}

The next step is to establish the existence of machines with 2-random totality probability.
For this we are going to apply Lemmas \ref{w1HmrXQxym} and \ref{qEkDLqgp25} 
and the fact from Sections \ref{pm7oworwdR} and \ref{lkNQieYqH} that 
the measure of any member of a universal \ml test
relative to $\emptyset'$ is a 2-random 
\lcep real. 

\begin{lem}[Machines with 2-random totality and infinitude probability]\label{UhzhB2hWw2}
There exist an oracle  machine $M$ and a monotone machine $N$
such that the measures of $\mathtt{TOT}(M)$, $\mathtt{INF}(M)$ and $\mathtt{INF}(N)$ are all
2-random \rcep reals.
\end{lem}
\begin{proof}
Let $U$ be a member of a universal \ml test relative to $\emptyset'$. Consider the machine $M$
of Lemma  \ref{w1HmrXQxym} with respect to $U$.
Since $\mathtt{TOT}(M)$ is a $\Pi^0_2$ class, 
by Lemma \ref{TLBE5A9gU4} its measure is a \rcep real.
Moreover $\mathtt{TOT}(M)$ is exactly the complement of the reals that have a prefix in $U$.
The latter is 2-random, as $U$ is a member of a universal \ml test relative to $\emptyset'$. Hence
$\mathtt{TOT}(M)$ is a 2-random real.
A similar application  of Lemma \ref{qEkDLqgp25} to the class $U$ shows that the constructed monotone
machine $N$ is such that 
$\mathtt{INF}(N)$ is a 2-random real.
\end{proof}

We are ready to prove the `only if' direction of Theorem \ref{YdyxkMqxRX}.
\begin{lem}[Randomness of totality and infinitude probabilities for universal oracle machines]\label{yH3aERIEO3}
If $U$ is a universal oracle machine then the measures of $\mathtt{TOT}(U)$ and $\mathtt{INF}(U)$ 
are both 2-random \rcep reals.
\end{lem}
\begin{proof}
Consider the machine $M$ of Lemma \ref{UhzhB2hWw2}, so that 
$\mu(\mathtt{TOT}(M))$ is a 2-random \rcep real. Since $U$ is universal, there exists
a string $\tau$ such that for all strings $\sigma$ we have $M(\sigma)\simeq U(\tau\ast\sigma)$.
We have
\[
\mathtt{TOT}(U)=\tau\ast\mathtt{TOT}(M)\cup \Big(\mathtt{TOT}(U)\cap (2^{\omega}-\dbra{\tau}])\Big)
\hspace{0.5cm}\textrm{and} \hspace{0.5cm}
\tau\ast\mathtt{TOT}(M)\cap \Big(\mathtt{TOT}(U)\cap (2^{\omega}-\dbra{\tau})\Big)=\emptyset. 
\]
Let $P=\mathtt{TOT}(U)\cap (2^{\omega}-\dbra{\tau})$ and note that this is a $\Pi^0_2$ class, so
$\mu(P)$ is a \rcep real by Lemma \ref{TLBE5A9gU4}.
So $\mu(\mathtt{TOT}(U))=2^{-|\tau|}\cdot \mu(\mathtt{TOT}(M))+\mu(P)$ is a 2-random \rcep real, as it is the sum of a
2-random \rcep real  and another \rcep real. A similar argument applies to 
$\mu(\mathtt{INF}(U))$.
\end{proof}

A very similar argument provides the `only if' direction of the part of
Theorem \ref{7MGSCRK2Yz} which refers to monotone machines.

\begin{lem}[Randomness of infinitude probabilities for universal monotone machines]
If $U$ is a universal monotone machine then the measure of $\mathtt{INF}(U)$ 
is a 2-random $\emptyset'$-right-c.e.\ real.
\end{lem}
\begin{proof}
We consider the monotone machine $N$ of 
Lemma \ref{qEkDLqgp25} and a universal monotone machine $U$. Then the proof is
the same as that of Lemma \ref{yH3aERIEO3}, with $M$ replaced by $N$ and 
$\mathtt{TOT}$ replaced by $\mathtt{INF}$.
\end{proof}

\subsection{From random reals to totality and infinitude probabilities}\label{VvLYmV4IPS}
In this section we prove the `if' direction of Theorem \ref{YdyxkMqxRX}
and the `if' direction of the part of 
Theorem \ref{7MGSCRK2Yz} which refers to monotone machines.
In other words, given a 2-random \lcep real $\alpha$ we will build a universal machine $U$
of the desired type such that the measure of $\mathtt{TOT}(U)$ or 
the measure of $\mathtt{INF}(U)$ is equal to $\alpha$.

\begin{lem}[Oracle machines]\label{DwQkO1LCfh}
If $\alpha\in (0,1)$ is a \rcep real and $c\in\Nat$ is such that $\alpha+2^{-c}<1$ 
then there exists a  machine $M$ and a string $\rho$ of length $c$ such that
$M(\sigma,n)$ is undefined for any $n$ and any 
string $\sigma$ which is compatible with $\rho$,
and $\mu(\mathtt{TOT}(M))=\mu(\mathtt{INF}(M))=\alpha$.
\end{lem}
\begin{proof}
According to the hypothesis, $1-\alpha-2^{-c}$ is a \lcep real in $(0,1)$, 
so there exists a $\emptyset'$-computable 
sequence $(b_i)$
of positive integers such that $1-\alpha-2^{-c}=\sum_i 2^{-b_i}$. 
So by the Kraft-Chaitin theorem relative to $\emptyset'$
there exists
a $\emptyset'$-computable sequence of strings $(\sigma_i)$ such that
$|\sigma_0|=c$, $|\sigma_{i+1}|=b_i$ for each $i$ and 
$S:=\{\sigma_i\ |\ i\in\Nat\}$ is a $\Sigma^0_2$ \pf set of strings.
Define $\rho:=\sigma_0$ and note that $\mu(\dbra{S})=1-\alpha$.
Then we can choose a canonical $\Sigma^0_2$ approximation $(S_i)$
to $S$ such that $\rho\in S_i$ for all $i$.
Then we can
apply the construction of Lemma \ref{w1HmrXQxym} to $S$ with this specific canonical 
$\Sigma^0_2$ approximation $(S_i)$ 
and we obtain a machine $M$ such that
$M(\sigma,n)$ is not defined for any string $\sigma$ compatible (with respect to the prefix relation) 
with $\rho$ and any $n$.
Then we get that
\[
\mu\big(\mathtt{TOT}(M)\big)=\mu\big(\mathtt{INF}(M)\big)=\mu\big(2^{\omega}-\dbra{S}\big)
=1-\mu\big(\dbra{S}\big)=\alpha,
\]
which concludes the proof.
\end{proof}
Lemma \ref{DwQkO1LCfh} has an analogue for monotone machines.

\begin{lem}[Monotone machines]\label{DODxKGG3H6}
If $\alpha\in (0,1)$ is a \rcep real and $c\in\Nat$ is such that $\alpha+2^{-c}<1$ 
then there exists a monotone machine $M$ and a string $\rho$ of length $c$ such that
$M(\sigma)$ is undefined for any  
string $\sigma$ which is compatible with $\rho$,
and $\mu(\mathtt{INF}(M))=\alpha$.
\end{lem}
\begin{proof}
The proof is the same as the proof of Lemma \ref{DwQkO1LCfh} only that
instead of obtaining the machine $M$ from the construction of  Lemma \ref{w1HmrXQxym}
we use Lemma \ref{UhzhB2hWw2}, which gives a monotone machine.
The rest of the parameters, including the $\Sigma^0_2$ set $S$ and its canonical 
$\Sigma^0_2$  approximation $(S_i)$, remain the same and as before
\[
\mu\big(\mathtt{INF}(M)\big)=\mu\big(2^{\omega}-\dbra{S}\big)
=1-\mu\big(\dbra{S}\big)=\alpha,
\]
where $M$ is now the monotone machine from
Lemma \ref{UhzhB2hWw2} based on $(S_i)$. This concludes the proof.
\end{proof}

We are ready to prove the remaining clauses of Theorems \ref{YdyxkMqxRX} and \ref{7MGSCRK2Yz}
which refer to oracle machines and monotone machines, as well as the statement immediately
below Theorems \ref{YdyxkMqxRX} which refers to the infinity probability of a universal 
oracle  machine.

\begin{lem}\label{aw4QkldwOa}
Let $\alpha\in (0,1)$ be a 2-random \rcep real. Then there exists
a universal oracle  machine $M$ such that 
$\mu(\mathtt{TOT}(M))=\alpha$.
Similarly, there exists
a universal oracle  machine $U$ such that 
$\mu(\mathtt{INF}(U))=\alpha$.
Finally there exists a monotone machine $T$ such that 
$\mu(\mathtt{INF}(T))=\alpha$.
\end{lem}
\begin{proof}
Let $V$ be a universal oracle machine and let $\gamma$ be the measure of the reals
$X$ such that
$V(X)$ is total. By Lemma \ref{TLBE5A9gU4} the real  $\gamma$ is a \rcep real.
By \eqref{5zlG2HCdb} and the discussion of Section \ref{lkNQieYqH} there exists $c\in\Nat$
such that $\alpha+2^{-c}<1$ and
the real $\beta:=\alpha - 2^{-c}\gamma$ is a \rcep real.
By Lemma \ref{DwQkO1LCfh} consider an oracle machine $N$ and a string $\rho$
of length $c$ such that 
the measure of $\mathtt{TOT}(N)$ is $\beta$ and $N(\sigma,n)$ is not defined for any
$\sigma$ which is compatible with $\rho$ and any $n$. Define an oracle machine $M$ as follows.
For each string $\sigma$ which is incompatible with $\rho$ and any $n$ let 
$M(\sigma,n)\simeq N(\sigma,n)$.
Moreover for each $\tau$ and any $n$ let $M(\rho\ast\tau,n)\simeq V(\tau,n)$. Since $V$ is
universal, it follows that $M$ is also a universal oracle machine. Moreover
\[
\mathtt{TOT}(M)=\rho\ast\mathtt{TOT}(V)\cup \mathtt{TOT}(N)
\hspace{0.5cm}\textrm{and}\hspace{0.5cm}
\rho\ast\mathtt{TOT}(V)\cap \mathtt{TOT}(N)=\emptyset
\]
and so
\[
\mu(\mathtt{TOT}(M))=2^{-|\rho|}\cdot\mu(\mathtt{TOT}(V)) + \mu(\mathtt{TOT}(N))
=2^{-c}\cdot \gamma+\beta =\alpha,
\]
which concludes the proof of the first clause. The argument for the second clause is entirely similar.
We let $\delta$ be the measure of the reals $X$ such that the domain of $V(X)$ is infinite
and choose $c$ such that
such that $\alpha+2^{-c}<1$ and
the real $\zeta:=\alpha - 2^{-c}\delta$ is a \rcep real. 
Then by Lemma \ref{DwQkO1LCfh} we can choose an oracle 
machine $F$ and a string $\eta$
of length $c$ such that 
the measure of $\mathtt{INF}(F)$ is $\beta$ and $F(\sigma,n)$ is not defined for any
$\sigma$ which is compatible with $\eta$ and any $n$. Then we define an
oracle  machine $U$ as follows.
For each string $\sigma$ which is incompatible with $\rho$ and any $n$ let 
$U(\sigma,n)\simeq F(\sigma,n)$.
Moreover for each $\tau$ and any $n$ let $U(\eta\ast\tau,n)\simeq V(\tau,n)$. Since $V$ is
universal, it follows that $U$ is also a universal oracle machine. Moreover
as before
\[
\mu(\mathtt{INF}(U))=2^{-|\eta|}\cdot\mu(\mathtt{INF}(V)) + \mu(\mathtt{INF}(F))
=2^{-c}\cdot \delta+\zeta =\alpha,
\]
which concludes the proof of the second clause of the lemma.
The proof of the third clause is entire analogous with the previous two, with the only difference that
we use Lemma \ref{DODxKGG3H6} instead of Lemma \ref{DwQkO1LCfh}.
\end{proof}

\subsection{Randomness of cofiniteness and computability probabilities}
In this section we prove the `only if' direction of Theorem \ref{Sgl7O4jKzk}. The methodology is
similar to the one we employed in Section \ref{6iJldN1dR6}, but the arguments will be more
involved since we are dealing with $\Sigma^0_3$ properties, namely cofiniteness and computability.
The first step is coding an arbitrary $\Sigma^0_3$ class into $\cofin{M}$ and $\comp{M}$ 
for some oracle machine $M$.
The proof of the following lemma draws considerably from the proof of the well-known fact that
the index set of the cofinite \ce sets
and the index set of the computable \ce sets are both $\Sigma^0_3$-complete.

\begin{lem}[from a $\Sigma^0_3$ set to an oracle machine for cofiniteness]\label{nl9d72crTM}
Given any upward closed 
$\Sigma^0_3$ set of stings $J$ there exists an oracle  machine $M$ such that the following
are equivalent for each $X$:
\begin{itemize}
\item $X$ has a prefix in $J$;
\item the domain of $M(X)$ is cofinite; and
\item the domain of $M(X)$ is computable.
\end{itemize}
Moreover, the domain of $M(X)$ is equal to the range of $M(X)$.
\end{lem}
\begin{proof}
Let $J$ be an upward closed $\Sigma^0_3$ set of strings. 
Then there exists a computable predicate $H_0$ such that 
$\sigma\in J$ is equivalent to $\exists t\forall n\exists s\ H_0(t,n,s,\sigma)$.
Consider the relation $\exists i\leq t\ \forall n\exists s\ H_0(i,n,s,\sigma)$. This is a $\Pi^0_2$ relation since
bounded quantifiers do not increase arithmetical complexity. So there exists a computable predicate
$H_1$ such that $\exists i\leq t\ \forall n\exists s\ H_0(t,n,s,\sigma)$ is equivalent to
$\forall n\exists s\ H_1(t,n,s,\sigma)$.
Moreover for each $q<p$, if $\forall n\exists s\ H_1(q,n,s,\sigma)$ then 
$\forall n\exists s\ H_1(p,n,s,\sigma)$.
We claim that there exists a computable predicate $H$ such that for each $\sigma,t$,
\[
\forall n\ \exists s\ H(t,n,s,\sigma)\iff 
\bigvee_{i=0}^{|\sigma|} \forall n \exists s\ H_1(t,n,s,\sigma\restr_i).
\]
This holds because arithmetical complexity above the first level of the arithmetical hierarchy is invariant under bounded disjunctions.
By the choice of $H_0,H_1$ and the fact that $J$ is upward closed,  
the predicate $H$ has the following properties:
\begin{enumerate}[\hspace{0.5cm}(a)]
\item  for each $q<p$ and each $\sigma$, if $\forall n\exists s\ H(q,n,s,\sigma)$ then 
$\forall n\exists s\ H(p,n,s,\sigma)$;
\item for each $q$ and each $\sigma\prec\tau$, if 
$\forall n\exists s\ H(q,n,s,\sigma)$ then
$\forall n\exists s\ H(q,n,s,\tau)$;
\item $\sigma\in J$ if and only if 
$\exists t\ \forall n\ \exists s\ H(t,n,s,\sigma)$.
\end{enumerate}

Let $(W_e)$ be a universal enumeration of all \ce sets. Then there exist
a computable function $g$ such that 
\begin{enumerate}[\hspace{0.5cm}(1)]
\item $\forall n\exists s\ H(t,n,s,\sigma)\iff |W_{g(t,\sigma)}|=\infty$; and
\item if $\sigma\preceq\tau$, $t\leq k$, and $|W_{g(t,\sigma)}|=\infty$, then
$|W_{g(k,\tau)}|=\infty$ 
\end{enumerate}
for all $t,k,\sigma,\tau$.
The second clause above says that 
the double sequence $(W_{g(t,\sigma)})$ is monotone, in the sense that
once a term is infinite, all later terms will also be infinite.
Let $(\sigma,n)\mapsto\tuple{\sigma,n}$ be a computable bijection 
from ordered pairs of strings and numbers $(\sigma,n)$ onto $\Nat$.
Moreover define $\Nat^{[\sigma]}=\{\tuple{\sigma,n}\ |\ n\in\Nat\}$ and
let $(\emptyset'_s)$ be a
computable enumeration  of the halting
problem $\emptyset'$.

For each binary string $\sigma$, we define a movable marker $m(\sigma)[s]$
dynamically to take values in $\Nat^{[\sigma]}$.
In the case where $m(\sigma)[s]$ reaches a limit as $s\to\infty$, that limit is denoted by
$m(\sigma)$.

At stage $0$ we let $m(\sigma)[0]=\tuple{\sigma,0}$ for all $\sigma$. At stage $s+1$,
\begin{itemize}
\item if  $\big|W_{g(|\sigma|,\sigma)}[s+1]\big|> \big|W_{g(|\sigma|,\sigma)}[s]\big|$ then let 
$m(\sigma)[s+1]=\tuple{\sigma,s+1}$;
\item if $|\sigma|\in\emptyset'_{s+1}-\emptyset'_{s}$, let $m(\sigma)[s+1]=\tuple{\sigma,s+1}$;
\item otherwise,  let $m(\sigma)[s+1]=m(\sigma)[s]$.
\end{itemize}
Note that for each $\sigma$, the marker $m(\sigma)[s]$ reaches a limit if and only if 
$W_{g(|\sigma|,\sigma)}$ is infinite.
We are ready to define the oracle  machine $M$ by induction on the stages $s$,
based on the function $g$.

Given any stage $s+1$, any $\sigma$ of length at most $s+1$ and any 
$n\in\Nat^{[\sigma]}$ such that $n<s+1$ we define
$M(\sigma, n)[s+1]=n$ for all $n<s+1$ such that $n\neq m(\sigma)[s+1]$, and
leave $M(\sigma, m(\sigma))[s+1]$ undefined.

By the definition of $M$ we have:
\begin{itemize}
\item if $m(\sigma)$ reaches a limit then $M(\sigma, m(\sigma))$ is undefined and 
$M(\sigma, i)\de$ for all $i\in \Nat^{[\sigma]}-\{m(\sigma)\}$;
\item if $m(\sigma)\to\infty$ then 
$M(\sigma, i)\de$ for all $i\in \Nat^{[\sigma]}$.
\end{itemize}
Let $X$ be a real. If $X$ has a prefix in $J$ then
by the properties of $H$ and $g$, the marker $m(X\restr_n)$ diverges for any sufficiently large $n$.
If, on the other hand, $X$ does not have a prefix in $J$, then for the same reasons
 $m(X\restr_n)$ reaches a limit for all $n$. Hence by the construction of $M$ and its properties,
 as discussed above, we have
\begin{itemize}
\item if $X$ has a prefix in $J$, then the domain of $M(X)$ is cofinite;
\item otherwise, the domain of $M(X)$ is coinfinite.
\end{itemize}
Moreover, by construction, the domain of $M(X)$ is equal to its range.
In addition, for each real $X$, if
the domain of $M(X)$ is not cofinite, 
then $m(X\restr_n)$ reaches a limit for each $n$, so the domain of $M(X)$ 
computes the halting problem $\emptyset'$.
This is because the final position of $m(X\restr_n)$ corresponds to the 
$n$th zero in the characteristic
sequence of the domain of $M(X)$. Hence
the settling time of $\emptyset'(n)$ is bounded above by 
the settling time of  $m(X\restr_n)$, which can
be calculated by the domain of $M(X)$. 
Hence if $X$ does not have a prefix in $J$, then 
the domain of $M(X)$ computes the halting problem, so it is not computable.
If, on the other hand, 
$X$ does have a prefix in $J$, then the domain of $M(X)$ is cofinite, hence computable.
\end{proof}

The above result allows us to construct a machine whose cofiniteness and computability probability
are both 3-random. This will later be used in order to show this property for any universal machine.

\begin{lem}[Machines with 3-random cofiniteness and computability probabilities]\label{2k8iykfZJP}
There exists an oracle machine $M$ 
such that $\cofin{M}$ and $\comp{M}$ are equal and have measure a
3-random \lcepp real.
\end{lem}
\begin{proof}
Let $U$ be a member of a universal \ml test relative to $\mathbf{0}^{(2)}$ 
and consider the machine $M$
of Lemma \ref{nl9d72crTM} with respect to $U$.
Then $\cofin{M}=\comp{M}$ and since this is a $\Sigma^0_3$ class,
by Lemma \ref{TLBE5A9gU4} its measure is a \lcepp real.
Moreover, by Lemma \ref{nl9d72crTM} 
the measure of $\cofin{M}$ is equal to the measure of $U$, which is 3-random,
as $U$ is a member of a universal \ml test relative to $\mathbf{0}^{(2)}$.
\end{proof}
Finally we are ready to prove the `only if' direction of Theorem \ref{Sgl7O4jKzk}.

\begin{lem}[Randomness of cofiniteness and computability probabilities for universal machines]
If $U$ is a universal oracle machine then the measures of $\mathtt{COF}(U)$ and $\mathtt{COM}(U)$ 
are both 3-random $\mathbf{0}^{(2)}$-left-c.e.\ reals.
\end{lem}
\begin{proof}
Consider the machine $M$ of Lemma \ref{2k8iykfZJP}, so that $\cofin{M}=\mu(\comp{M})$ and
$\mu(\cofin{M})$ is a 
3-random $\mathbf{0}^{(2)}$-left-c.e.\ real. Since $U$ is universal, there exists
a string $\tau$ such that for all strings $\sigma$ we have $M(\sigma)\simeq U(\tau\ast\sigma)$.
We have
\[
\mathtt{COF}(U)=\tau\ast\mathtt{COF}(M)\cup \Big(\mathtt{COF}(U)\cap (2^{\omega}-\dbra{\tau})\Big)
\hspace{0.5cm}\textrm{and} \hspace{0.5cm}
\tau\ast\mathtt{COF}(M)\cap \Big(\mathtt{COF}(U)\cap (2^{\omega}-\dbra{\tau})\Big)=\emptyset 
\]
Let $P=\mathtt{COF}(U)\cap (2^{\omega}-\dbra{\tau})$ and note that this is a $\Sigma^0_3$ class, so
$\mu(P)$ is a $\mathbf{0}^{(2)}$-left-c.e.\ real by Lemma \ref{TLBE5A9gU4}.
So $\mu(\mathtt{COF}(U))=2^{-|\tau|}\cdot \mu(\mathtt{COF}(M))+\mu(P)$ 
is a 3-random $\mathbf{0}^{(2)}$-left-c.e.\ real as the sum of a
3-random $\mathbf{0}^{(2)}$-left-c.e.\ real  and another $\mathbf{0}^{(2)}$-left-c.e.\ real. 
\end{proof}

\subsection{From random reals to cofiniteness and computability probabilities}\label{Vavq5ziBpO}
In this section we prove the `if' direction of Theorem \ref{Sgl7O4jKzk}.
We first show how to obtain a machine with prescribed cofiniteness probability.
\begin{lem}[Prescribed cofiniteness probability]\label{OXrsKZCS9}
If $\alpha\in (0,1)$ is a  $\mathbf{0}^{(2)}$-left-c.e.\ real 
and $c\in\Nat$ is such that $\alpha+2^{-c}<1$ 
then there exists an oracle machine $M$ and a string $\rho$ of length $c$ such that
$M(\sigma,n)$ is undefined for any $n$ and any 
string $\sigma$ which is compatible with $\rho$,
and $\mu(\mathtt{COF}(M))=\alpha$.
\end{lem}
\begin{proof}
Let $(b_i)$ be a $\mathbf{0}^{(2)}$-computable 
sequence of positive integers such that $\alpha=\sum_i 2^{-b_i}$. 
By the Kraft-Chaitin theorem relative to $\mathbf{0}^{(2)}$
and since $\alpha+2^{-c}<1$, there exists
a $\mathbf{0}^{(2)}$-computable sequence of strings $(\sigma_i)$ such that
$|\sigma_0|=c$, $|\sigma_{i+1}|=b_i$ for each $i$ and 
$S:=\{\sigma_{i+1}\ |\ i\in\Nat\}$ is a $\Sigma^0_3$ \pf set of strings.
Define $\rho:=\sigma_0$ and note that $\mu(\dbra{S})=\alpha$.
Let $J$ be the upward closure of $S$ and construct an oracle
machine $M$ as in the proof of Lemma \ref{nl9d72crTM} based on the $\Sigma^0_3$ upward
closed set of strings $J$, with the additional restriction that $M(\sigma,n)$ is not
defined for any $\sigma$ which is compatible with $\rho$ and any $n$.
Then by the same arguments, all of the properties of $M(X)$ that are listed in the proof of 
Lemma \ref{nl9d72crTM} hold as long as $X$ is not prefixed by $\rho$.
If $\rho$ is a prefix of $X$, then $M(X)$ is the empty function, so its domain is coinfinite.
Hence
\[
\mu\big(\mathtt{COF}(M)\big)=
\mu\big(\mathtt{COF}(M)\cap (2^{\omega}-\dbra{\rho})\big)
+\mu\big(\mathtt{COF}(M)\cap \dbra{\rho}\big)=
\mu(\dbra{S})+0=\alpha
\]
which concludes the proof.
\end{proof}

A minor modification of the above argument gives a similar result regarding the computability
probability.
\begin{lem}[Prescribed computability probability]\label{2x3FJmdR7Y}
If $\alpha\in (0,1)$ is a  $\mathbf{0}^{(2)}$-left-c.e.\ real 
and $c\in\Nat$ is such that $2^{-c}<\alpha$ 
then there exists an oracle machine $M$ and a string $\rho$ of length $c$ such that
$M(\sigma,n)$ is undefined for any $n$ and any 
string $\sigma$ which is compatible with $\rho$,
and $\mu(\mathtt{COM}(M))=\alpha$.
\end{lem}
\begin{proof}
Let $(b_i)$ be a $\mathbf{0}^{(2)}$-computable 
sequence of positive integers such that $\alpha-2^{-c}=\sum_i 2^{-b_i}$. 
By the Kraft-Chaitin theorem relative to $\mathbf{0}^{(2)}$, there exists
a $\mathbf{0}^{(2)}$-computable sequence of strings $(\sigma_i)$ such that
$|\sigma_0|=c$, $|\sigma_{i+1}|=b_i$ for each $i$ and 
$S:=\{\sigma_{i+1}\ |\ i\in\Nat\}$ is a $\Sigma^0_3$ \pf set of strings.
Define $\rho:=\sigma_0$ and note that $\mu(\dbra{S})=\alpha-2^{-c}$.
Let $J$ be the upward closure of $S$ and construct an oracle
machine $M$ as in the proof of Lemma \ref{nl9d72crTM} based on the $\Sigma^0_3$ upward
closed set of strings $J$, with the additional restriction that $M(\sigma,n)$ is not
defined for any $\sigma$ which is compatible with $\rho$ and any $n$.
Then by the same arguments, all of the properties of $M(X)$ that are listed in the proof of 
Lemma \ref{nl9d72crTM} hold as long as $X$ is not prefixed by $\rho$.
If $\rho$ is a prefix of $X$, then $M(X)$ is the empty function, so its domain is 
computable
Hence
\[
\mu\big(\mathtt{COM}(M)\big)=
\mu\big(\mathtt{COM}(M)\cap (2^{\omega}-\dbra{\rho})\big)
+\mu\big(\mathtt{COM}(M)\cap \dbra{\rho}\big)=
\mu(\dbra{S})+2^{-c}=\alpha
\]
which concludes the proof.
\end{proof}

We are now ready to prove the `if' direction of Theorem \ref{Sgl7O4jKzk}.
We use Lemma \ref{OXrsKZCS9}  and Lemma \ref{2x3FJmdR7Y}
in order to produce  universal oracle machines with prescribed probabilities.

\begin{lem}
Let $\alpha\in (0,1)$ be a 3-random $\mathbf{0}^{(2)}$-left-c.e.\ real. Then there exist
universal oracle  machines $M$ and $N$ such that 
$\mu(\mathtt{COF}(M))=\alpha=\mu(\mathtt{COM}(N))$.
\end{lem}
\begin{proof}
Let $V$ be a universal oracle machine and let $\gamma=\mu(\cofin{V})$. 
By Lemma \ref{TLBE5A9gU4} the real  $\gamma$ is a $\mathbf{0}^{(2)}$-left-c.e.\ real.
By \eqref{5zlG2HCdb} and the discussion of Section \ref{lkNQieYqH} there exists $c\in\Nat$
such that $\alpha+2^{-c}<1$ and
the real $\beta:=\alpha - 2^{-c}\gamma$ is a $\mathbf{0}^{(2)}$-left-c.e.\ real.
By Lemma \ref{OXrsKZCS9} consider an oracle machine $F$ and a string $\rho$
of length $c$ such that 
the measure of $\mathtt{COF}(F)$ is $\beta$ and $F(\sigma,n)$ is not defined for any
$\sigma$ which is compatible with $\rho$ and any $n$. 
Define an oracle machine $M$ as follows.
For each string $\sigma$ which is incompatible with $\rho$ and any $n$ let 
$M(\sigma,n)\simeq F(\sigma,n)$.
Moreover for each $\tau$ and any $n$ let $M(\rho\ast\tau,n)\simeq V(\tau,n)$. Since $V$ is
universal, it follows that $M$ is also a universal oracle  machine. Moreover,
\[
\mathtt{COF}(M)=\rho\ast\mathtt{COF}(V)\cup \mathtt{COF}(F)
\hspace{0.5cm}\textrm{and}\hspace{0.5cm}
\rho\ast\mathtt{COF}(V)\cap \mathtt{COF}(F)=\emptyset,
\]
and so
\[
\mu(\mathtt{COF}(M))=2^{-|\rho|}\cdot\mu(\mathtt{COF}(V)) + \mu(\mathtt{COF}(F))
=2^{-c}\cdot \gamma+\beta =\alpha
\]
which concludes the proof of the first equality.

For the second equality, we let $\delta$ be the measure of the reals $X$ such that the domain of $V(X)$ is computable
and choose $c$ such that
such that $\alpha+2^{-c}<1$ and
the real $\alpha - 2^{-c}\delta$ is a \rcep real. Now let $\zeta = \alpha - 2^{-c} \delta + 2^{-c}$.  
Then by Lemma \ref{2x3FJmdR7Y}, we can choose an oracle 
 machine $G$ and a string $\eta$
of length $c$ such that 
the measure of $\mathtt{COM}(G)$ is $\beta$ and $G(\sigma,n)$ is not defined for any
$\sigma$ which is compatible with $\eta$ and any $n$. Then we define an
oracle  machine $N$ as follows.
For each string $\sigma$ which is incompatible with $\eta$ and any $n$ let 
$U(\sigma,n)\simeq G(\sigma,n)$.
Moreover for each $\tau$ and any $n$ let $U(\eta\ast\tau,n)\simeq V(\tau,n)$. Since $V$ is
universal, it follows that $U$ is also a universal oracle  machine.  Moreover
\[
\mu(\mathtt{COM}(N))=2^{-|\eta|}\cdot\mu(\mathtt{COM}(V)) + \mu(\mathtt{COM}(G)) - 2^{-c}
=2^{-c}\cdot \delta+\zeta - 2^{-c} =\alpha.
\]
which concludes the proof of the second equality. 
\end{proof}

\subsection{Outline of the proof of Theorem \ref{7FXYwDy4zY}}
A classic fact from computability theory says that
the property that $W_e$ (the $e$th \ce set) is Turing complete, is $\Sigma^0_4$-complete.
Theorem \ref{7FXYwDy4zY} is a version of this fact
in terms of measures of oracle Turing machines.		

\begin{repthm}{7FXYwDy4zY}		
The probability that a universal oracle machine computes a function whose domain computes the		
halting problem is		
a 4-random \lceppp real. Conversely, every 4-random \lceppp real in $(0,1)$ is the probability that		
a certain universal oracle machine computes a function whose domain computes the		
halting problem.		
\end{repthm}

At this point we have given enough arguments in order to illustrate the methodology
of obtaining characterizations of probabilities of universal oracle machines in terms of 
algorithmic randomness. For this reason, we give a mere outline of the
proof of Theorem \ref{7FXYwDy4zY}, which is entirely along the lines of the proof of
Theorem \ref{7MGSCRK2Yz} and Theorem \ref{Sgl7O4jKzk}
which were presented in excruciating detail.
The crucial ingredient is the following lemma, which allows the construction of oracle
machines with suitably prescribed probability of computing the halting problem.
\begin{lem}[Domain or range computing the halting problem]\label{1zgPV8YMh}
Given a $\Sigma^0_4$ upward closed set of strings, 
there exists an oracle  machine $M$ such that for each 
$X$ the domain (or range) of $M(X)$ 
computes the halting problem if and only if $X$ has a prefix in $J$.
\end{lem}
The proof of \ref{1zgPV8YMh} is based on the classic argument showing that
the index set of the complete \ce sets is $\Sigma^0_4$-complete 
(\eg see \cite[Corollary XII 1.7]{MR882921}).
Given this result and using a member of a universal \ml test relative to $\mathbf{0}^{(3)}$ as a
$\Sigma^{0}_4$ set of strings $J$, we can show that the completeness probability of
a universal oracle machine is a 4-random  $\mathbf{0}^{(3)}$-left-\ce real.
For the converse one follows faithfully the structure of the argument that we 
developed in Section \ref{Vavq5ziBpO}. We use the relativized Kraft-Chaitin method along with 
\eqref{5zlG2HCdb} and the discussion of Section \ref{lkNQieYqH}
in order to construct a universal oracle machine  whose probability of computing a function with
domain computing the halting problem equals
a given 4-random $\mathbf{0}^{(3)}$-left-\ce real.

\subsection{Probabilities for  oracle and monotone machines that are not so random}\label{rT1udReGjP}
In this section we show that the closely related cofiniteness probability
considered in Proposition \ref{XqR4J6nv5M0}
is not as random as one might expect.
Given an oracle machine $M$, consider the class of reals such that $M(X)$ is a total 
function which is the characteristic sequence of a cofinite set. Then this 
is the intersection of a $\Pi^0_2$ class, imposing totality, and a $\Sigma^0_2$ class. Indeed,
this class is equal to
\[
\Big\{X\ |\ \forall n\exists s\ M(X,n)[s]\de\Big\}\cap 
\Big\{X\ |\ \exists t\ \forall s\ \forall n>t\  \big(M(X,n)[s]\un\ \vee\ M(X,n)[s]\de=1\big)\Big\}
\]
where clearly the first is a $\Pi^0_2$ class and the second is a $\Sigma^0_2$ class.
Therefore the measure of this class
of reals is the difference of two \lcep reals.
Rettinger (see \cite{Zheng2004,Rettinger2005} or \cite[Theorem 9.2.4 ]{rodenisbook}) proved that such a real is either
 \lcep or \rcep or it is not 2-random. In any of these cases, the probability that 
$M(X)$ is a total computable
function is not a 3-random real.

The same argument holds for the case of monotone machines where we look
at the outcome that the output is a computable stream or a stream with a tail of 1s.
We have thus proved Proposition \ref{XqR4J6nv5M0}.

\section{Probabilities of infinitary self-delimiting machines}\label{geZv34fjP6}
In this section we prove all the results that refer to 
infinitary self-delimiting machines. We devote Section  \ref{ELZgFIxGxP} to
defining and discussing this special machine model. We emphasize that although the general methodology 
is the same as the one we developed in Section \ref{oQfiYDwopG}, the arguments are special
to the rather different nature of the model of infinitary self-delimiting machines.

\subsection{Infinitary self-delimiting machines}\label{ELZgFIxGxP}
We define Chaitin's infinitary self-delimiting machines which were originally introduced in 
Chaitin \cite{CHAITIN1976233} and later studied in
\cite{DBLP:journals/jsyml/BecherG05,jsyml/BecherFGM06,tcs/BecherG07,DBLP:journals/jsyml/BecherG09}
and \cite{ndjfl/BecherFNP05}. This model is based on a partial computable function 
$M: 2^{<\omega}\times\Nat\to2^{<\omega}$ 
with certain properties, which can be used to define a self-delimiting infinitary model
$M^{\infty}$, which in turn can be seen as a \pf infinitary model 
$M^{\ast}$. 

The properties we require for the partial computable function $M$ are:
\begin{enumerate}[\hspace{0.5cm}(a)]
\item if $M(\sigma,m)\de$ and $n<m$, then $M(\sigma,n)\de$ and $M(\sigma,n)\preceq M(\sigma,m)$;
\item if $M(\sigma,n)\de$, then  for all strings $\tau$, $M(\sigma\ast\tau,n)\de$ and  
$M(\sigma\ast\tau,n)=M(\sigma,n)$;
\item the relation $M(\sigma,n)\de$ is decidable.
\end{enumerate}
These conditions may seem non-standard from the point of view of classical computability theory.
However they can be understood if one considers the following definition of infinitary self-delimiting
computations that they facilitate.
\begin{defi}[Infinitary self-delimiting machines]\label{7osrvJOuxR}
If $M: 2^{<\omega}\times\Nat\to2^{<\omega}$ 
is a partial computable function
which satisfies conditions
(a)-(c) above, define the infinitary machine $M^{\infty}$ as follows:
\begin{enumerate}[\hspace{0.5cm}(i)]
\item $M^{\infty}(\sigma)\de$ if $M(\sigma,n)\de$ for all $n$;
\item If $M^{\infty}(\sigma)\de$, then 
$M^{\infty}(\sigma)$ is the supremum of all $M(\sigma,n), n\in\Nat$
\end{enumerate}
where $\sigma$ is a binary string and $n\in\Nat$. Note that $M^{\infty}(\sigma)$ could be a string or a stream.
\end{defi}

We can now see how the properties (a)-(c) on $M$ 
above give a self-delimiting quality to $M^{\infty}$.
Indeed, the properties on $M$ imply that 
\begin{equation*}
\parbox{13cm}{if $M^{\infty}(\sigma)\de$ then for any $\tau$ we have $M^{\infty}(\sigma\ast\tau)\de$ and the two outputs are equal.}
\end{equation*}
Note that if $M^{\infty}(\sigma)\de$ and it is finite, \ie a string, then 
$M^{\infty}(\sigma\ast\tau)$ is not allowed to be a proper extension of 
$M^{\infty}(\sigma)\de$. This is one of the main differences with the standard oracle Turing machine model
and the monotone machines.
If $M^{\infty}(\sigma)\de$ then the values of $M^{\infty}(\tau)$ for any extension $\tau$ of $\sigma$
are completely determined by $M^{\infty}(\sigma)$. 
The other difference is that the relation $M^{\infty}(\sigma)\de$ is \pz and in general not $\Sigma^0_1$ as
we might be used to in classical computability theory.
The domain of $M^{\infty}$ is denoted by $\mathtt{DOM}(M^{\infty})$ and consists of the strings $\sigma$ 
such that $M^{\infty}(\sigma)\de$.

In the same way that we view self-delimiting machines as \pf machines, we can do the same
with these infinitary self-delimiting machines.
Given $M$ and $M^{\infty}$ which is defined in terms of $M$, we can define the relation 
$M^{\ast}(\sigma)\de$ to mean that $\sigma$ is a minimal string such that $M^{\infty}(\sigma)\de$
(\ie  $M^{\infty}(\sigma)\de$ and  $M^{\infty}(\tau)\un$ for all proper prefixes of $\sigma$).
Moreover if $M^{\ast}(\sigma)\de$ then let $M^{\ast}(\sigma):=M^{\infty}(\sigma)\de$.
Then the domain of $M^{\ast}$,
denoted by $\mathtt{DOM}(M^{\ast})$ consists of the strings $\sigma$ 
such that $M^{\ast}(\sigma)\de$ and is a \pf set of strings. In this way, 
$M^{\ast}$ may be regarded as a  \pf machine of a higher type. 
Formally, $M^{\ast}: 2^{<\omega}\to 2^{<\omega}$ and
\[
M^{\ast}(\sigma)=\begin{cases}
M^{\infty}(\sigma), & \textrm{if for every $\tau\in 2^{<\omega}$, $M^{\infty}(\sigma)\de=M^{\infty}(\sigma\ast\tau)\de$}\\
\un, & \textrm{otherwise.}
\end{cases}
\]
Note that $\mathtt{DOM}(M^{\infty})$ is a \pz set of strings. On the other hand, 
$\mathtt{DOM}(M^{\ast})$ is merely $\Delta^0_2$
since to decide membership in $\mathtt{DOM}(M^{\ast})$ one has to ask a finite number of
 $\Pi^0_1$ questions, namely a finite number of questions about membership in 
 $\mathtt{DOM}(M^{\infty})$.
Note further that $\mathtt{DOM}(M^{\ast})\subseteq \mathtt{DOM}(M^{\infty})$ and the
corresponding $\Sigma^0_1$ classes are equal, namely
$\dbra{\mathtt{DOM}(M^{\infty})}=\dbra{\mathtt{DOM}(M^{\ast})}$, so that
\[
\mu\Big(\mathtt{DOM}(M^{\ast})\Big)=\mu\Big(\mathtt{DOM}(M^{\infty})\Big)=
\sum_{\sigma\in \mathtt{DOM}(M^{\ast})} 2^{-|\sigma|}.
\]
The transition from $M^{\infty}$ to $M^{\ast}$ is  trivial, so we will mainly work with
$M^{\infty}$ in the following. However, the results we present concerning 
$M^{\infty}$ also apply to $M^{\ast}$.
Universality for $M^{\infty}$ is defined as in most machine models.
Note that there is an effective list of all partial computable functions
$M: 2^{<\omega}\times\Nat\to2^{<\omega}$ with the properties (a)-(c) above. Such a list induces
an effective list $(M_e^{\infty})$ of all infinitary self-delimiting machines.
The relation $\simeq$ denotes the fact that either the expressions on either side of it are undefined
(or do not halt) or both of these expressions are defined and are equal. 

\begin{defi}[Universal infinitary self-delimiting machines]
Given an effective list $(M^{\infty}_e)$ of all infinitary self-delimiting machines, an 
infinitary self-delimiting machine $U^{\infty}$ is universal if 
there exists a computable function $e\mapsto\sigma_e$ such that
$U^{\infty}(\sigma_e\ast\tau)\simeq M^{\infty}_e(\tau)$
for all $e, \tau$.
\end{defi}

Theorem \ref{7MGSCRK2Yz} 
concerns the question of whether the outcome of a universal 
infinitary self-delimiting machine that runs on a {\em random} input is finite or
infinite. 
We conclude the present section with an analysis of these outcomes from a complexity point of view.
This discussion will be the basis for the arguments of Section \ref{geZv34fjP6}.
Given an infinitary self-delimiting machine $M^{\infty}$  we define:
\begin{itemize}
\item $\mathtt{INF}(M^{\infty})$ is the set of strings in $\mathtt{DOM}(M^{\infty})$ such that 
$M^{\infty}(\sigma)$ is a stream;
\item $\mathtt{FIN}(M^{\infty})=\mathtt{DOM}(M^{\infty})-\mathtt{INF}(M^{\infty})$.
\end{itemize}

We note that $\mathtt{INF}(M^{\infty})$
is a $\Pi^0_2$ set because to determine membership in $\mathtt{INF}(M^{\infty})$
it is enough to first decide membership
$\mathtt{DOM}(M^{\infty})$ and then 
ask a $\Pi^0_2$ question about whether the length of the output
is infinite or not. Similarly,  
 $\mathtt{FIN}(M^{\infty})$ is a $\Sigma^0_2$ set.
Moreover, by standard manipulations of quantifiers we have that
 $\mathtt{INF}(M^{\ast})$
is a $\Pi^0_2$ set and $\mathtt{FIN}(M^{\ast})$ is a $\Sigma^0_2$ set.
The reader can easily verify that these complexity bounds are the best possible in general
and exhibit machines such that these sets are complete for the arithmetical classes that they belong to.
Similar remarks can be made about the classes of reals represented by these sets of strings
(recall the notation $\dbra{S}$ for a set of strings $S$ from Section \ref{pm7oworwdR}).
For example, $\dbra{\mathtt{DOM}(M^{\infty})}$ is not only an open set but also a $\Sigma^0_1(\emptyset')$
class, and the same is true of $\dbra{\mathtt{FIN}(M^{\infty})}$, which of course equals 
$\dbra{\mathtt{FIN}(M^{\ast})}$.
Moreover $\dbra{\mathtt{INF}(M^{\infty})}$
is also an open set but merely a $\Sigma^0_1(\emptyset'')$ class since membership in
$\mathtt{INF}(M^{\infty})$ requires oracle $\emptyset''$.
Also note that $\dbra{\mathtt{INF}(M^{\infty})}$ is the difference of two 
$\Sigma^0_1(\emptyset')$ classes, namely 
$\dbra{\mathtt{DOM}(M^{\infty})}-\dbra{\mathtt{FIN}(M^{\infty})}$.
Similar remarks apply to the measures of these sets.
For example, 
the measures of $\mathtt{DOM}(M^{\infty})$, $\mathtt{FIN}(M^{\ast})$ are \lcep reals
while the measure of $\mathtt{INF}(M^{\infty})$ is the difference of two \lcep reals.
Let $\cofin{M^{\infty}}$ consist of the strings in 
$\dom{M^{\infty}}$ such that $M^{\infty}(\sigma)$
is equal to $\tau\ast 1^{\omega}$ for some string $\tau$.

\subsection{Randomness of finiteness probability}
In this section we prove the part of Theorem \ref{7MGSCRK2Yz} which refers to infinitary
self-delimiting machines. 
We emphasize that this is the content of Becher, Daicz, and Chaitin \cite{firstBC}. Our proof
here serves one main purpose: it facilitates the proof of the converse, which is one of our results
regarding infinitary self-delimiting machines. In addition, our proof is different and more concise than
in \cite{firstBC}, being based on concepts of definability rather than initial segment complexity.

Due to the nature of this model, we are not able to prove
Lemma \ref{w1HmrXQxym} for infinitary
self-delimiting machines. Instead, we show a slightly weaker version, namely that
given any upward closed $\Sigma^0_2$ set $U$ of strings
we can define a machine $M^{\infty}$ such that for each string $\sigma$ we have $M^{\infty}(\sigma)\de$
if and only if $\sigma$ has a (finite) covering consisting of strings in $U$. 
In terms of classes of reals, this means
that $M^{\infty}(\sigma)\de$ if and only if $\sigma\in \dbra{U}$.
In the following we let $\mathtt{EXT}(\sigma,n)$ denote the set of extensions of string $\sigma$
of length $n$.

\begin{lem}\label{tU86O3w2du}
Given an upward closed $\Sigma^0_2$ set $U$ of strings there exists a 
machine $M$ which never prints anything on its tape and such that, for all $\sigma$, 
\begin{enumerate}[\hspace{1cm}(i)]
\item $\sigma\in U\Rightarrow \exists n\geq |\sigma|\ \forall \tau\in 
\mathtt{EXT}(\sigma,n),\ \tau\in  \mathtt{DOM}(M^{\infty})$
\item $\sigma\in \mathtt{DOM}(M^{\infty})\Rightarrow  \sigma\in U$.
\end{enumerate}
	
\end{lem}
\begin{proof}
Let $(U_s)$ be a canonical
$\Sigma^0_2$ approximation of $U$, as defined in
Section \ref{pm7oworwdR}.
We define $(\sigma,n)\mapsto M(\sigma,n)$ computably in stages.
Let $M(\sigma,0)$ be defined for all $\sigma$ and equal to the empty string.
At stage $s+1$ we first define $M(\sigma,s+1)$ for all $\sigma$ of length at most
$s$ and then define $M(\tau,s+1)$ for all $\tau$ of length $s+1$.
Therefore by the end of stage $s$, we have determined exactly all $M(\tau,t)$
for all $\tau$ of length at most $s$ and all $t\leq s$. Let $s\in\Nat$ and suppose 
inductively that the first $s$ stages of the construction have been completed. At
stage $s+1$, and each string $\sigma$ of length at most $s$ we first check
if $M(\sigma,s)\un$ and in that case we let $M(\sigma,s+1)\un$. Otherwise we check if
there exists a prefix of $\sigma$ in  $U_{s+1}$. In this case we let $M(\sigma,s+1)\de$ and equal to
the empty string. Otherwise we let $M(\sigma,s)\un$. For each string $\tau$ of length $s+1$ 
we perform a similar check.
We check if $\tau$ has a prefix in  $U_{s+1}$, and in that case we define
$M(\tau,t)\de=\lambda$ for each $t\leq s+1$, where $\lambda$ denotes
the empty string. Otherwise we let
$M(\tau,t)\un$ for all $t\leq s+1$. This completes stages $s+1$ and the inductive definition of $M$.

By a straightforward induction on the steps of the construction we have
\begin{enumerate}[\hspace{0.5cm}(1)]
\item $M$ has the properties (a)-(c) of Section \ref{ELZgFIxGxP}
and it never prints anything on the output tape;
\item for each $\tau,s$ we have $M(\tau,s)\de$ if and only if 
for each $|\tau|\leq t\leq s$
there exists $\rho\preceq\tau$
such that $\rho\in U_t$.
\end{enumerate}
We can now use property (2) in order to establish property (i) of the lemma.
Suppose that $\sigma\in U$. Then there exists a least stage $s$ such that for all $t\geq s$ 
we have $\sigma\in U_t$. Then by property (2), for all strings $\tau\in\extp{\sigma,s}$
we have $M(\tau,n)\de$ for all $n\in\Nat$. In other words, if $\sigma\in U$ then there exists some
$s$ such that for all  $\tau\in\extp{\sigma,s}$ we have $M^{\infty}(\tau)\de$.
For (ii), assume that  $\sigma\not\in U$. Then there certainly exists $s$ such that $\sigma \notin U_{s+1}$ and 
hence $\tau \notin U_{s+1}$ for any prefix $\tau$ of $\sigma$. It follows from the definition of $M$ that $M(\sigma,s+1) \uparrow$ 
and therefore $\sigma \notin \mathtt{DOM}(M^{\infty})$. This   completes the proof of the lemma.
\end{proof}

Becher, Daicz, and Chaitin \cite{firstBC} called a program $\sigma$ of $M^{\infty}$
{\em circular} if $M^{\infty}(\sigma)\de$ and $M^{\infty}(\sigma)$ is finite, \ie a string.
Note that all the programs in the domain of $M^{\infty}$ of Lemma \ref{tU86O3w2du} 
are circular.
This shows that the $\Sigma^0_2$ complexity is not hidden specifically in 
$\mathtt{FIN}(M^{\infty})$ but rather in the $\Pi^0_1$ definition of convergence of an infinitary
self-delimiting machine, \ie in the domain $\mathtt{DOM}(M^{\infty})$ itself. In particular, as a direct
consequence of Lemma \ref{tU86O3w2du} we have
\begin{equation}\label{d4p3siEeVR}
\parbox{11cm}{Given an upward closed $\Sigma^0_2$ set $S$ of strings there exists a 
machine $M$ such that $\mathtt{INF}(M^{\infty})$ is empty and
$\dbra{\mathtt{DOM}(M^{\infty})}=\dbra{S}$.}
\end{equation}
We remark that
the proof of  Lemma \ref{tU86O3w2du} can be modified in a straightforward way
so that the constructed machine
$M^{\infty}$ has output $M^{\infty}(\sigma)=0^{\omega}$ whenever $M^{\infty}(\sigma)\de$.
This modified construction shows the dual of \eqref{d4p3siEeVR}, namely that
given an upward closed $\Sigma^0_2$ set $S$ of strings there exists a 
machine $M$ such that $\mathtt{FIN}(M^{\infty})$ is empty and
$\dbra{\mathtt{DOM}(M^{\infty})}=\dbra{S}$. 

The next step is to establish the existence of a machine $M$ such that 
$\mathtt{FIN}(M^{\infty})$ is 2-random.

\begin{lem}\label{mwJDegabOa}
There exists an infinitary self-delimiting machine $M^{\infty}$ such that 
$\mathtt{INF}(M^{\infty})$ is empty and
the measure of 
$\mathtt{DOM}(M^{\infty})=\mathtt{FIN}(M^{\infty})$ is a 2-random \lcep real. 
\end{lem}
\begin{proof}
Consider the member $S$ of a \ml test relative to $\emptyset'$. This can be represented 
as an upward closed $\Sigma^0_2$ set of strings, and its measure is 2-random.
Therefore the statement follows by a direct application of \eqref{d4p3siEeVR} to $S$.
\end{proof}

Finally we are ready to deduce the `only if' direction of 
the part of Theorem \ref{7MGSCRK2Yz} which refers to infinitary
self-delimiting machines, which was originally proved by
Becher, Daicz, and Chaitin \cite{firstBC}.

\begin{lem}[Universal finiteness probability]
If $U^{\infty}$ is a universal infinitary
self-delimiting machine then $\mathtt{FIN}(M^{\infty})$ is a 2-random \lcep real. 
\end{lem}
\begin{proof}
Consider the machine $M^{\infty}$ of Lemma \ref{mwJDegabOa}, so that 
$\mu(\mathtt{FIN}(M))$ is a 2-random \lcep real. Since $U^{\infty}$ is universal, there exists
a string $\tau$ such that for all strings $\sigma$ we have 
$M^{\infty}(\sigma)\simeq U^{\infty}(\tau\ast\sigma)$.
We have
\[
\mathtt{FIN}(U^{\infty})=\tau\ast\mathtt{FIN}(M^{\infty})\cup 
\Big(\mathtt{FIN}(U^{\infty})\cap (2^{\omega}-\dbra{\tau}])\Big)
\hspace{0.2cm}\textrm{and} \hspace{0.2cm}
\tau\ast\mathtt{FIN}(M^{\infty})\cap \Big(\mathtt{FIN}(U^{\infty})\cap 
(2^{\omega}-\dbra{\tau})\Big)=\emptyset 
\]
Let $P=\mathtt{FIN}(U)\cap (2^{\omega}-\dbra{\tau})$ and note that this is a $\Sigma^0_2$ class, so
$\mu(P)$ is a \lcep real by Lemma \ref{TLBE5A9gU4}.
So $\mu(\mathtt{FIN}(U))=2^{-|\tau|}\cdot \mu(\mathtt{FIN}(M))+\mu(P)$ 
is a 2-random \lcep real as the sum of a
2-random \lcep real  and another \lcep real.
\end{proof}

\subsection{From random numbers to the finiteness probability}
In this section we follow the methodology which were 
developed in Section \ref{VvLYmV4IPS} for a different machine model.
We see that this is straightforward, given Lemma \ref{tU86O3w2du}
and the fact that $\dom{M^{\infty}}$ and $\fin{M^{\infty}}$ are  
$\Sigma^0_2$ sets.

\begin{lem}\label{izUAz38qL6}
If $\alpha<1$ is a \lcep real and $2^{-c}<1-\alpha$, 
then there exists an infinitary self-delimiting machine 
$M^{\infty}$ and a string $\rho$ of length $c$ such that 
$M^{\infty}(\sigma)$ is undefined for any string $\sigma$ which is a prefix or a suffix of $\rho$,
$\dom{M^{\infty}}=\fin{M^{\infty}}$, and the measure of the domain of $M^{\infty}$ is  
$\alpha$.
\end{lem}
\begin{proof}
Given $\alpha$ and $c$, let $(\alpha_s)$ be a $\emptyset'$-computable 
increasing sequence of rationals converging to $\alpha$ with $\alpha_0=0$. 
We apply the  Kraft-Chaitin algorithm with sequence of requests
$2^{-c}$, $\alpha_1-\alpha_0$, $\alpha_2-\alpha_1, \dots$ and we get
a $\emptyset'$-computable enumeration of a \pf set of strings whose measure is
$2^{-c}+\alpha$. Let $\rho$ be the first of these strings and let $U$ contain the rest of them.
Then we may choose a canonical $\Sigma^0_2$ approximation $(U_s)$ to $U$ such that no string that
is compatible with $\rho$ belongs to $U_s$ for any $s$. Then we can
apply Lemma \ref{tU86O3w2du} to $U$ with this specific 
$\Sigma^0_2$ approximation $(U_s)$ 
and we obtain a machine $M$ such that
$M^{\infty}(\sigma)$ is not defined for any string $\sigma$ compatible with $\rho$.
Moreover we get that the measure of $\dom{M^{\infty}}$ is equal to the measure of 
$U$, which is $\alpha$.
\end{proof}

We are now ready to deduce the `if' direction of 
the part of Theorem \ref{7MGSCRK2Yz} which refers to infinitary
self-delimiting machines. Note that this is the converse of the result of
Becher, Daicz, and Chaitin \cite{firstBC}.

\begin{lem}
Let $\alpha$ be a 2-random \lcep real. Then there exists
a universal infinitary self-delimiting machine $M^{\infty}$ 
such that $\mu(\dom{M^{\infty}})=\alpha$.
\end{lem}
\begin{proof}
Let $V^{\infty}$ be a universal infinitary self-delimiting machine 
and let $\gamma=\mu(\dom{V^{\infty}})$. 
Then  $\gamma$ is a \lcep real.
Let $c$ be a constant such that the real $\beta=\alpha - 2^{-c}\gamma$ is a \lcep real
and $\beta+2^{-c}$ is less than 1.
By Lemma \ref{izUAz38qL6}, consider an 
infinitary self-delimiting machine $N^{\infty}$ and a string $\rho$
of length $c$ such that 
the measure of $\dom{N^{\infty}}$ is $\beta$ and $\dom{N^{\infty}}$ does not contain any
string which is compatible with $\rho$. Define an 
infinitary self-delimiting machine $M^{\infty}$ as follows.
For each string $\sigma$ which is incompatible with $\rho$ let $M^{\infty}(\sigma)\simeq V(\sigma)$.
Moreover for each $\tau$ let $M^{\infty}(\rho\ast\tau)\simeq N(\tau)$. Then
\[
\dom{M^{\infty}}=\dom{V^{\infty}}\cup \rho\ast\dom{N^{\infty}}
\hspace{0.5cm}\textrm{and}\hspace{0.5cm}
\dom{V^{\infty}}\cap \rho\ast\dom{N^{\infty}}=\emptyset
\]
and so
\[
\mu(\dom{M^{\infty}})=\mu(\dom{V^{\infty}}) + 2^{-|\rho|}\cdot\mu(\dom{N^{\infty}})
=\beta+2^{-c}\cdot \gamma =\alpha
\]
which concludes the argument.
\end{proof}

\subsection{Restricted infinite models and higher randomness restrictions}
The model of infinitary self-delimiting machines that we described in Section 
\ref{ELZgFIxGxP} was introduced in 
Chaitin \cite{CHAITIN1976233} and
studied in Becher, Daicz and Chaitin \cite{firstBC}.
Although one can exhibit 2-random probabilities in this model, there is a fundamental reason
why higher randomness is not attainable as the measure of a subset of the domain of such a machine.
The reason for this is that the domain can be seen as a \pf $\Delta^0_2$ set of strings.
The following lemma can be used in order to give a formal proof of this fact.

\begin{lem}\label{SoFOwQ5ys}
There is no \ce \pf set of strings that contains a 
$\Sigma^0_1(\emptyset')$ subset of 2-random measure.
More generally, 
for each $n\in\Nat$ and any  $\Sigma^0_{n+1}$ \pf set of strings there is no 
$\Sigma^0_{n+2}$ subset of this set which has $(n+2)$-random measure.
\end{lem}

\begin{proof}
Let $S$ be a c.e.\ \pf set of strings and let $V$ be a $\Sigma^0_1(\emptyset')$ subset of $S$.
We construct a \ml test $(U_n)$ relative to $\emptyset'$ such that $\mu(V) \in\cap_n U_n$.
Note that  $\mu(V)$ is a \lcep real. 
Given $n$ we use $\emptyset'$ to compute a finite subset $D_n$ of $S$
such that $\mu(S)-\mu(D_n)<2^{-n-1}$.
Then let $\epsilon_n$ be a rational which is less than $2^{-|\rho|}$ for any $\rho\in D_n$.
Let $\delta_n=2^{-n-1}\cdot\epsilon_n/(1+\epsilon_n)$ and note that the number $d_n$ of strings in $D_n$ is
at most $\mu(S)/\epsilon_n\leq 1/\epsilon_n$. Moreover note that $\delta_n<\epsilon_n$.
Let $(\tau_i)$ be a $\emptyset'$-computable enumeration of $V$ and let
$V_s=\{\tau_i\ |\ i< s\}$. For each $n,s$ define
\[
J(n,s)=(\mu(V_s), \mu(V_s)+\delta_n)
\]
and let $U_n=\cup_s J(n,s)$. Since the intervals $J(n,s)$ have fixed length $\delta_n$ it follows that
$\mu(V)\in U_n$ for each $n$. Moreover the intervals $J(n,s)$ are uniformly computable in $\emptyset'$
so $(U_n)$ are uniformly $\Sigma^0_1$ in $n$. It remains to verify that $\mu(U_n)\leq 2^{-n}$ for each $n$.

Now $U_n$ may be decomposed as the union of $J(n,0)$ together with $\cup_s (J(n,s+1) \setminus J(n,s)$). We have  
$\mu(J(n,0) = \delta_n$ and, for each $s$, $\mu(J(n,s+1) \setminus J(n,s)) = \min\{\delta_n, 2^{-|\tau_s|}\}$.
The values of $s$ may be considered in two cases. First, there are $d_n$ values for which $\tau_s \in D_n$; together with 
the measure $\delta_n$ from $J(n,0)$, this gives measure $\leq (1+d_n) \delta_n \leq 2^{-n-1}$, since $d_n \leq 1/\epsilon_n$. 
Second, there are all the values of $s$ for which $\tau_s \notin D_n$. Here the total measure is  
$\sum_{\tau_s \in (V \setminus D_n)} 2^{-|\tau_s|} \leq \mu(S) - \mu(D_n) \leq 2^{-n-1}$. It follows that $\mu(U_n) \leq 2^{-n}$. as desired.
Hence $(U_n)$ is a \ml test relative to $\emptyset'$ and
$\mu(V)\in U_n$ for all $n$, which means that $\mu(V)$ is not 2-random.
\end{proof}

Let us discuss the implications of this result. In the following informal discussion 
we implicitly assume that 
{\em probability of a property of a machine} with a notion of a {\em domain} refers to the measure
of a subset of its domain.
By the first clause of 
Lemma \ref{SoFOwQ5ys}, any \pf machine model where the relation of convergence 
(namely the domain) is 
$\Sigma^0_1$  cannot exhibit properties whose probability is a second-order $\Omega$ number,
\ie a 2-random \lcep real. This holds because the measure of a $\Sigma^0_2$ \pf set of strings is
always a \lcep real.
Chaitin's model for infinite computations, as we discussed, can be seen as a higher type \pf machine,
where the domain is a \pz \pf set of strings. This non-standard infinitary feature allowed the presence of
properties that occur with 2-random probability, as we saw in the previous sections. However
the second clause of Lemma \ref{SoFOwQ5ys} for $n=1$ says that this model cannot exhibit
properties whose probability is a 3-random \lcepp real. Indeed, convergence in this model is \pz and so it is
$\Sigma^0_2$.
This observation gives a limit to the complexity of relative $\Omega$ numbers we can exhibit in this model, this limit
being $\Omega$ numbers of the second level, \ie relative to the second iteration of the halting problem. 

Let $\cofin{M^{\infty}}$ be the set of $\sigma$ in the domain of $M^{\infty}$ such that 
$M^{\infty}(\sigma)$ is a stream with a tail of 1s.
Note that $\cofin{M^{\infty}}$ is a $\Sigma^0_3$ set of strings, and it is not hard to see
that it is  $\Sigma^0_3$-complete whenever $M^{\infty}$ is universal.
In particular, the measure of 
$\cofin{M^{\infty}}$ is always a \lcepp real. So the above observations imply that for
every infinitary self-delimiting machine 
$M^{\infty}$ the measure of $\cofin{M^{\infty}}$ is never a 3-random real.
Similarly, Lemma \ref{SoFOwQ5ys} implies
Lemma \ref{5JCUWCcstk}.

Becher and Chaitin \cite{fuin/BecherC02} study a self-delimiting model for infinite computations 
which is close to the model of \cite{firstBC}, but which avoids the above problem, as the
convergence notion is now $\Pi^0_2$.
Using this modified model,  Becher and Chaitin exhibit probabilities that are random relative to $\emptyset''$.

\end{document}